\documentclass[cmp]{svjour}
\usepackage{amssymb}
\usepackage{amsfonts}
\usepackage{graphicx}
\usepackage{color}
\usepackage{times}
\usepackage{mathptmx} % assumes new font selection scheme installed
\usepackage{datetime}
\usepackage{framed}
\usepackage{graphics}

\newtheorem{lem}{Lemma}
\newtheorem{thm}{Theorem}

\def\fa{\mathfrak{a}}
\def\fb{\mathfrak{b}}

\def\<{\leqslant}           % nice less than or equal to sign
\def\>{\geqslant}           % nice larger than or equal to sign
\def\div{\mathrm{div}}         % divergence
\def\d{\partial}

\def\wt{\widetilde}

\def\Re{\mathrm{Re}}   % real part
\def\Im{\mathrm{Im}}   % imaginary part

\def\cH{\mathcal{H}}   % Hardy space
\def\mR{\mathbb{R}}    % real line
\def\mC{\mathbb{C}}    % complex plane

\def\Tr{\mathrm{Tr}}       % matrix trace
\def\rT{\mathrm{T}}        % matrix transpose
\def\bS{\mathbf{S}}

\def\bE{\mathbf{E}}    % expectation

\def\[[[{[\![\![}
\def\]]]{]\!]\!]}

\def\bra{\langle}
\def\ket{\rangle}

\def\Bra{\left\langle}
\def\Ket{\right\rangle}

\def\re{\mathrm{e}}        % number e
\def\rd{\mathrm{d}}        % differential

\def\cL{\mathcal{L}}

\def\br{\mathbf{r}}
\def\x{\times}
\def\ox{\otimes}

\def\fA{\mathfrak{A}}
\def\fB{\mathfrak{B}}

\def\fF{\mathfrak{F}}
\def\fG{\mathfrak{G}}

\def\sM{\mathsf{M}}
\def\sK{\mathsf{K}}

\def\cF{\mathcal{F}}
\def\cW{\mathcal{W}}

\def\mF{\mathbb{F}}

\def\cI{\mathcal{I}}

\def\eps{\epsilon}
\def\Ups{\Upsilon}
\def\ups{\upsilon}

\def\sn{|\!|\!|}

\begin{document}
\title{Invariant states of linear quantum stochastic systems under Weyl perturbations of the Hamiltonian and coupling operators\thanks{This work is supported by the Air Force Office of Scientific Research (AFOSR) under agreement number FA2386-16-1-4065.}}
\author{Igor G. Vladimirov%\inst{1}
\and \quad Ian R. Petersen%\inst{1}
\and \quad Matthew R. James%\inst{1}% etc
}                     % Do not remove
\institute{Research School of Engineering, College of Engineering and Computer Science, Australian National University, ACT 2601, Canberra, Australia. E-mail: igor.g.vladimirov@gmail.com, i.r.petersen@gmail.com, matthew.james@anu.edu.au}

\date{Received: date / Accepted: date}
% Add name of the expert who has communicated your paper
\communicated{[name]}
\maketitle

\begin{abstract}
This paper is concerned with the sensitivity of invariant states in linear quantum stochastic systems with respect to nonlinear perturbations. The system variables are governed by a Markovian  Hudson-Parthasarathy  quantum stochastic differential equation (QSDE) driven by quantum Wiener processes of external bosonic fields in the vacuum state. The quadratic system Hamiltonian and the linear system-field coupling operators, corresponding to a nominal open quantum harmonic oscillator, are subject to perturbations represented in a Weyl quantization form.  Assuming that the nominal linear QSDE has a Hurwitz dynamics matrix and using the Wigner-Moyal phase-space framework, we carry out an infinitesimal perturbation analysis of
the quasi-characteristic function for the  invariant  quantum state of the nonlinear  perturbed system. The resulting correction of the invariant states in the spatial frequency domain  may find applications to their approximate computation, analysis of relaxation dynamics and non-Gaussian state generation in nonlinear quantum stochastic systems.
\end{abstract}

\keywords{
Quantum stochastic system, invariant quantum state, open quantum harmonic oscillator, Weyl quantization, perturbation analysis.
}

\medskip
\noindent
{\bf MSC 2010 codes.}
Primary:
81S22, % Open systems, reduced dynamics, master equations, decoherence
81S25, % Quantum stochastic calculus
81S30, % Phase-space methods including Wigner distributions, etc.
81P16, % Quantum state spaces, operational and probabilistic concepts
81S05; %Canonical quantization, commutation relations and statistics
secondary:
81Q15, % Perturbation theories for operators and differential equations
35Q40, %PDEs in connection with quantum mechanics
37M25. % Computational methods for ergodic theory (approximation of invariant measures, computation of Lyapunov exponents, entropy

%
%%%%%%%%%%%%%%%%%%%%%%%%%%%%%%%%%%%%%%%%%%%%%%%%%%%%%%%%%%%%%%%%%%%%%%%%%%%%%%%%%%%%%%%%%%%%%%%%%%%
\section{Introduction}\label{sec:intro}
%%%%%%%%%%%%%%%%%%%%%%%%%%%%%%%%%%%%%%%%%%%%%%%%%%%%%%%%%%%%%%%%%%%%%%%%%%%%%%%%%%%%%%%%%%%%%%%%%%%

At the core of emerging quantum computation and quantum information technologies is the employment of hidden resources of physical systems at the atomic scales governed by the laws of quantum mechanics. In contrast to the  classical Newtonian mechanics, such systems are described by operator-valued quantities acting on an underlying Hilbert space which evolve according to unitary similarity transformations when the system is isolated from the environment. This unitary evolution is specified by the Hamiltonian which represents self-energy of the system.  When the system is coupled to the surroundings in the form of interaction with a macroscopic measuring apparatus, the internal state of the system is disturbed in a random fashion, depending on what is measured,  which makes noncommuting operators inaccessible to simultaneous measurement \cite{H_2001,M_1998,S_1994}. This specific stochastic nature of quantum systems is described in terms of quantum probability which replaces the scalar-valued classical probability measures with density operators (quantum states) acting on the same Hilbert space as the dynamic variables.

An important part of quantum computation  protocols \cite{NC_2000} is the initialization of quantum systems in certain classes of states. Similarly to the classical terminal state control problem \cite{PBGM_1962},  the quantum state preparation is a quantum control problem which can be solved by varying the parameters of the Hamiltonian  either in an open-loop fashion or by using feedback \cite{DP_2010}. Another approach \cite{Y_2012,Y_2009} to the state preparation consists in letting the quantum plant of interest interact with external bosonic fields (such as nonclassical electromagnetic radiation) in a specifically arranged fashion, resulting in a dissipative  system  with a required invariant state. The interaction of an open quantum system with such fields is modelled by using the Hudson-Parthasarathy quantum stochastic calculus \cite{HP_1984,P_1992}. This approach
represents the open quantum dynamics in the form of quantum stochastic differential equations (QSDEs) driven by noncommutative counterparts of the classical Wiener process \cite{KS_1991}, with the quantum Wiener processes acting on the symmetric Fock spaces \cite{PS_1972}. In addition to the Hamiltonian, the energetics of the system-field interaction is specified by the coupling operators. In combination with scattering matrices (which take into account photon exchange between the fields in multichannel setups), this approach is used by the theory of quantum feedback networks \cite{GJ_2009,JG_2010} in order to model interconnections of open quantum systems which interact with each other and the environment.

The quantum state generation problem is important both for finite-level systems (such as qubit registers) and continuous variables systems which find application in different (for example, quantum-optical) platforms of quantum computing \cite{NC_2000}.  An important class of such systems is provided by open quantum harmonic oscillators (OQHOs) whose dynamic variables satisfy canonical commutation relations (CCRs) similar to those of the quantum mechanical positions and momenta \cite{M_1998,S_1994}, and    are governed by linear QSDEs. The Hamiltonian and coupling operators of the OQHO are quadratic and linear functions of the system variables. The dynamics of such systems, when they are driven by bosonic fields in the vacuum state,  are quite similar (at least  at the level of the mean values and covariances) to those of classical Gaussian Markov diffusion processes (for example, the Ornstein-Uhlenbeck process \cite{KS_1991}). In particular, if the OQHO is initialized in a Gaussian quantum state \cite{KRP_2010}, then its Gaussian nature is preserved in time. Furthermore, if the dynamics matrix of the linear QSDE is Hurwitz, the state of the system converges weakly \cite{B_1968,CH_1971} to a unique invariant Gaussian state regardless of whether the initial state is Gaussian. This property of OQHOs enables them to be used for generating  a Gaussian state by allowing the system to evolve over a sufficiently long period of time. The dissipation mechanism, which is built in this procedure, secures stability of the invariant state being achieved in the long run \cite{PAMGUJ_2014}.

The implementation of OQHOs (for example, using quantum optical components such as cavities, beam splitters and phase shifters) can involve parametric uncertainties, leading to Gaussian invariant states with different covariances \cite{VPJ_2017a}. More complicated inaccuracies result in deviations from the idealized linear-quadratic energetics. The appearance of nonquadratic terms in the Hamiltonian and nonlinearities in the coupling operators makes the governing QSDEs nonlinear and leads to non-Gaussian invariant quantum states. The deviation from the nominal invariant Gaussian quantum state can be studied in the framework of a particular description of the uncertainties in the energy operators.
Such uncertainties can be modeled as functions of the system variables. However, there are many ways to extend usual functions of several real or complex  variables to the noncommutative quantum variables. One of such extensions is provided by the Weyl functional calculus \cite{F_1989} which employs unitary Weyl operators whose role in this context is similar to that of the spatial Fourier harmonics. In fact, the Weyl quantization replaces these harmonics  in the Fourier integral representations of classical functions with the Weyl operators. In application to the energy operators of the quantum stochastic system, the Weyl quantization leads to partial differential equations (PDEs) or integro-differential equations (IDEs) for the evolution of the quasi-characteristic function (QCF) \cite{CH_1971}. The latter is a quantum counterpart of the usual characteristic functions of classical probability distributions and encode information on the mixed moments of quantum system variables.

The QCF and its Fourier transform --- the Wigner quasi-probability density function  (QPDF) --- constitute the basis of the Wigner-Moyal phase-space approach \cite{Hi_2010,M_1949} which avoids the ``burden of the Hilbert space'' by representing the quantum state dynamics in the more conventional form of PDEs and IDEs involving only real or complex variables and functions thereof. Although the Moyal equations  \cite{M_1949} for the QCF and QPDF dynamics were obtained originally  for isolated quantum systems, there also are extensions of the phase-space approach  to different classes of open quantum systems \cite{GRS_2014,GRS_2015,KS_2008,MD_2015,V_2015c}.

In the present paper, following \cite{V_2015c}, we apply the phase-space approach to the infinitesimal perturbation analysis of the invariant state of an OQHO when its linear-quadratic energy operators are subject to variations represented in the Weyl quantization form.
The effect of such perturbations of the Hamiltonian
 on the second moments of the system variables (with the coupling operators remaining unperturbed) was studied in \cite{SVP_2014}.
  In the case of linear system-field coupling,  the Weyl quantization of the Hamiltonian was used for approximate computation of invariant states via operator  splitting \cite{Mar_1988,S_1968} in \cite{V_2015c} and dissipation relations for a weighted $L^2$-norm of the QCF, similar to the norm in the Bessel potential space \cite{S_2008}, and for the $\chi^2$-divergence (the second-order Renyi relative entropy \cite{R_1961}) of the actual QPDF from the nominal Gaussian PDF. Similarly to and in furtherance of the previous developments, we make advantage of the tractability of the Gaussian state dynamics for linear quantum stochastic systems and the ease of computing the moments of the Weyl operators over Gaussian quantum states. This is combined with the machinery of linear PDEs and related semigroups \cite{E_1998,V_1971} together with  the generalized functions \cite{V_2002} in the spatial frequency domain in order to develop a perturbative (linear response) theory for QCFs of quantum systems subject to the Weyl variations both  in the Hamiltonian and coupling operators. More precisely,  we compute the influence functions which relate the first-order corrections of the QCF and the first two moments for the perturbed invariant quantum state to the strength functions of the Weyl variations and discuss norm bounds for these perturbation operators.

 We also mention that the Weyl quantization of the energy operators was applied to the measurement-based quantum stochastic filtering in \cite{V_2016,V_2017} and was used in \cite{V_2015a,V_2015b} in regard to optimality conditions for coherent (measurement-free) quantum control and filtering problems for linear quantum plants. The coherent quantum control and filtering by interconnection \cite{JNP_2008,MJ_2012,MP_2009,NJP_2009,P_2014,PH_2015,P_2016,VP_2013a,VP_2013b,VP_2016,ZJ_2011b}  provide a promising modern paradigm which has the potential to outperform the traditional observation-actuation approaches of the classical control theory.

The paper is organised as follows.
Section~\ref{sec:sys} specifies the class of quantum stochastic systems being considered.
Section~\ref{sec:IDE} describes the Weyl quantization of energy operators and the integro-differential equation for the evolution of the quasi-characteristic function.
Section~\ref{sec:gauss} discusses the dynamics of quantum states and preservation of their Gaussian nature for open quantum harmonic oscillators.
Section~\ref{sec:Weylvar} specifies the Weyl perturbations of the Hamiltonian and coupling operators of such an oscillator.
Section~\ref{sec:statevar} computes the corresponding first-order correction terms of the QCF and QPDF for the resulting perturbed system.
Section~\ref{sec:gain} discusses upper bounds for the
%influence functions which relate
norms of the linear operators relating
these corrections to the variations in the energy operators.
Section~\ref{sec:numer} gives numerical examples of the invariant state perturbation analysis, including the sensitivity bounds.
Section~\ref{sec:conc} provides a conclusion and outlines possible directions of research.

%%%%%%%%%%%%%%%%%%%%%%%%%%%%%%%%%%%%%%%%%%%%%%%%%%%%%%%%%%%%%%%%%%%%%%%%%%%%%%%%%%%%%%%%%%%%%%%%%%%
\section{Quantum stochastic dynamics}
\label{sec:sys}
%%%%%%%%%%%%%%%%%%%%%%%%%%%%%%%%%%%%%%%%%%%%%%%%%%%%%%%%%%%%%%%%%%%%%%%%%%%%%%%%%%%%%%%%%%%%%%%%%%%

We will be concerned with an open quantum system interacting with external fields as shown in Fig.~\ref{fig:sys}.
%==============================================================================
\begin{figure}[htbp]
\centering
\unitlength=0.8mm
\linethickness{0.2pt}
\begin{picture}(100.00,45.00)
    %\put(7.5,7.5){\dashbox(35,15)[cc]{}}
    \put(40,20){\framebox(20,20)[cc]{system}}

    \put(20,36){\vector(1,0){20}}
  %  \put(20,33){\vector(1,0){20}}
    \put(20,24){\vector(1,0){20}}
    \put(30,31.5){\makebox(0,0)[cc]{{$\vdots$}}}
    \put(15,31.5){\makebox(0,0)[cc]{{$\vdots$}}}

    \put(15,36){\makebox(0,0)[cc]{$W_1$}}
%    \put(15,33){\makebox(0,0)[cc]{$W_2$}}
    \put(15,24){\makebox(0,0)[cc]{$W_m$}}

    \put(60,36){\vector(1,0){20}}
  %  \put(60,33){\vector(1,0){20}}
    \put(60,24){\vector(1,0){20}}
    \put(70,31.5){\makebox(0,0)[cc]{{$\vdots$}}}
    \put(85,31.5){\makebox(0,0)[cc]{{$\vdots$}}}

    \put(85,36){\makebox(0,0)[cc]{$Y_1$}}
 %   \put(85,33){\makebox(0,0)[cc]{$Y_2$}}
    \put(85,24){\makebox(0,0)[cc]{$Y_m$}}
\end{picture}\vskip-15mm
\caption{A schematic depiction of the open quantum stochastic system interacting with the input quantum Wiener processes $W_1, \ldots, W_m$ and producing the output fields  $Y_1, \ldots, Y_m$.}
\label{fig:sys}
%\end{center}
\end{figure}
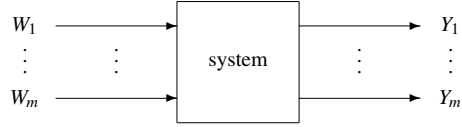
%==============================================================================
The system dynamics are described in terms of an even number $n$ of variables $X_1,\ldots, X_n$ assembled into a vector $X:= (X_k)_{1\< k\< n}$ (vectors are organized as columns). These   system variables are time-varying self-adjoint operators on a complex separable Hilbert space $\cH$. Their evolution  is governed by a Markovian Hudson-Parthasarathy quantum stochastic differential equation (QSDE) \cite{HP_1984,P_1992}
\begin{equation}
\label{dX}
    \rd X =
    \cL(X)\rd t  - i[X,h^{\rT}]\rd W
\end{equation}
whose structure is described below (we will often omit the time arguments for brevity).
Although it formally resembles classical SDEs \cite{KS_1991}, the QSDE (\ref{dX}) is driven by a vector $W:= (W_k)_{1\< k \< m}$ of an even number $m$ of self-adjoint operator-valued quantum Wiener processes $W_1, \ldots, W_m$ acting on a symmetric Fock space $\cF$. These  processes model the external bosonic fields \cite{H_1991,P_1992} and satisfy the quantum Ito relations
\begin{eqnarray}
\label{dWdW}
    \rd W\rd W^{\rT}
    & := &
    (\rd W_j\rd W_k)_{1\<j,k\< m}
    = \Omega \rd t,\\
\label{Omega}
        \Omega
    & := &
    I_m + iJ,\\
\label{J}
        J
        & := &
    \left[
    \begin{array}{cc}
    0 & 1\\
    -1 & 0
    \end{array}
    \right] \ox I_{m/2},
\end{eqnarray}
where $\ox$ is the Kronecker product of matrices, and $I_r$ is the identity matrix of order $r$.
In application to matrices of operators,  the transpose $(\cdot)^{\rT}$ acts as if their entries were scalars.
In contrast to the identity diffusion matrix of the standard Wiener process, $\Omega$ is  a complex positive semi-definite Hermitian matrix with an orthogonal antisymmetric imaginary part   $\Im \Omega=J$ (so that $J^2 =-I_m$). Therefore, the quantum Wiener processes  $W_1, \ldots, W_m$ do not commute with each other. Moreover, they have the following  two-point commutator matrix
\begin{eqnarray}
\nonumber
    [W(s), W(t)^{\rT}]
    & := &
    ([W_j(s), W_k(t)])_{1\<j,k\<m}\\
\label{WWst}
    & = &
    2i\min(s,t)J ,
    \qquad s,t\>0,
\end{eqnarray}
where $[\alpha, \beta]:= \alpha \beta - \beta \alpha$ is the commutator of linear operators $\alpha$ and $\beta$. In accordance with $W$ representing the $m$-channel input field, $h:= (h_k)_{1\< k\< m}$ in (\ref{dX}) is a vector of system-field coupling operators $h_1, \ldots, h_m$ which are also self-adjoint operators on the space $\cH$. Since the entries of the commutator matrix
$
    [X, h^{\rT}]
    :=    ([X_j,h_k])_{1\< j\< n,1\<k\< m}
$
are skew-Hermitian operators, the dispersion $(n\x m)$-matrix $-i [X, h^{\rT}]$ in (\ref{dX}) consists of self-adjoint operators on $\cH$.
The
drift vector
$    \cL(X)
$
of the QSDE (\ref{dX}) consists of $n$ operators and is obtained by the entrywise application of
the Gorini-Kossakowski-Sudar\-shan-Lindblad (GKSL) generator    \cite{GKS_1976,L_1976}, which acts on a single system operator $\xi$ as
\begin{equation}
\label{cL}
   \cL(\xi)
   := i[h_0,\xi]
     +
         \frac{1}{2}
    \big(
        [h^{\rT},\xi]\Omega h  + h^{\rT}\Omega [\xi,h]
    \big).
\end{equation}
and involves the system Hamiltonian $h_0$ (which is also a self-adjoint operator on $\cH$) in addition to the coupling operators $h_1, \ldots, h_m$. The energy operators $h_0, h_1, \ldots, h_m$ are appropriately extended functions of the system variables $X_1, \ldots, X_n$.
The superoperator $\cL$ in (\ref{dX})  is a quantum analogue of the infinitesimal generators of classical Markov diffusion processes \cite{KS_1991,S_2008}.
The specific structure of the drift vector and the dispersion matrix of the QSDE (\ref{dX}) comes from  the system-field interaction which drives a unitary operator $U(t)$ on the system-field tensor-product space $\cH:= \cH_0\ox \cF$ (with $\cH_0$ the initial space for the action of the system variables at time $t=0$) as
\begin{equation}\label{dU}
    \rd U(t) = -U(t) \Big(i(h_0(t)\rd t + h(t)^{\rT} \rd W(t)) + \frac{1}{2}h(t)^{\rT}\Omega h(t)\rd t\Big),
\end{equation}
where $U(0)=\cI_{\cH}$  is the identity operator on $\cH$. The unitary operator $U(t)$, which  is associated with the system-field interaction over the time interval from $0$ to $t$, is adapted in the sense that  it acts effectively on the subspace $\cH_0\ox \cF_t$, where $\{\cF_t:\, t\>0\}$ is the Fock space filtration.   The corresponding quantum stochastic flow at time $t$ involves a unitary similarity transformation, which acts  on operators $\zeta$ on $\cH$  as $ \zeta\mapsto U(t)^{\dagger} \zeta U(t)$ (and applies entrywise to vectors of operators) and evolves the system variables as
\begin{equation}
\label{uni}
    X(t)
    =
    U(t)^{\dagger} (X(0)\ox \cI_{\cF}) U(t),
\end{equation}
where $(\cdot)^{\dagger}$ is the operator adjoint.
The QSDE (\ref{dX}) can be  obtained from (\ref{uni}) by using (\ref{dU}) and the quantum Ito formula \cite{HP_1984,P_1992} in combination with (\ref{Omega}) and commutativity between the forward Ito increments $\rd W(t)$ and adapted processes taken at time $s\< t$. Adapted processes $\xi$, which are functions  of the system variables, satisfy QSDEs of the same form
\begin{equation}
\label{dxi}
    \rd \xi
    =
    \cL(\xi)\rd t - i[\xi,h^{\rT} ]\rd W.
\end{equation}
An alternative representation of the quantum Wiener process $W$ (which reflects photon absorption and emission and is particularly relevant to quantum optics)  is in terms of
the field annihilation  $\fa_1, \ldots, \fa_{m/2}$ and creation $\fa_1^{\dagger}, \ldots, \fa_{m/2}^{\dagger}$  processes \cite{HP_1984,P_1992}, assembled into vectors $\fa:= (\fa_k)_{1\<k\<m/2}$ and $\fa^{\#}:= (\fa_k^{\dagger})_{1\<k\<m/2}$:
\begin{eqnarray}
\nonumber
    W
    & := &
    2
    \left[
    \begin{array}{c}
        \Re \fa\\
        \Im \fa
    \end{array}
    \right]
    =
    \left[
    \begin{array}{c}
        \fa + \fa^{\#}\\
        i(\fa^{\#}-\fa)
    \end{array}
    \right]\\
\label{Wfa}
    & = &
    \left(
    \left[
    \begin{array}{cc}
        1 & 1\\
        -i & i
    \end{array}
    \right]
    \ox I_{m/2}
    \right)
    \breve{a}.
\end{eqnarray}
Here, $(\cdot)^{\#}$ denotes the entrywise operator adjoint (reducing to the complex conjugate in the case of complex matrices), and the real and imaginary parts
are extended to matrices $N$ with operator-valued entries as $\Re N = \frac{1}{2}(N+N^{\#})$ and $\Im N = \frac{1}{2i}(N-N^{\#})$ which consist of self-adjoint operators. Also,
\begin{equation}
\label{double}
    \breve{\fa}
    :=
    \left[
    \begin{array}{c}
        \fa \\
        \fa^{\#}
    \end{array}
    \right]
\end{equation}
is an auxiliary doubled-up vector.
In accordance with (\ref{dWdW}) and (\ref{Omega}),
the quantum Ito table of the annihilation and creation processes is represented in terms of (\ref{double}) as
\begin{eqnarray*}
    \rd
    \breve{a}
    \rd
    \breve{a}^{\dagger}
    & =
    \left[
    \begin{array}{cc}
        \rd\fa \rd \fa^{\dagger} & \rd \fa \rd \fa^{\rT}\\
        \rd\fa^{\#} \rd \fa^{\dagger} & \rd \fa^{\#} \rd \fa^{\rT}
    \end{array}
    \right]\\
    & =
    \left[
    \begin{array}{cc}
        1 & 0\\
        0 & 0
    \end{array}
    \right]
    \ox I_{m/2}
    \rd t,
\end{eqnarray*}
where $(\cdot)^{\dagger}:= ((\cdot)^{\#})^{\rT}$ applies to matrices of operators as the transpose of the entrywise adjoint and reduces to the complex conjugate transpose $(\cdot)^*:= (\overline{(\cdot)})^{\rT}$ for complex matrices. Also, the term
$$
    ih^{\rT}\rd W = L^{\dagger}\rd \fa - L^{\rT}\fa^{\#},
$$
which is part of the diffusion term in (\ref{dU}),
is related to a different vector $L:=(L_k)_{1\<k\<m/2}$ of (not necessarily self-adjoint) coupling operators $L_1, \ldots, L_{m/2}$, so that
\begin{equation}
\label{hL}
    h =
        -J
        \left[
        \begin{array}{c}
          \Re L \\
          \Im L
        \end{array}
        \right],
        \qquad
        \left[
        \begin{array}{c}
          \Re L \\
          \Im L
        \end{array}
        \right]
        =
        Jh.
\end{equation}
Here, use is also made of the property $J^2 = -I_m$ of the matrix $J$ in (\ref{Omega}) mentioned above.
The relations (\ref{Wfa}) and (\ref{hL}) allow the two alternative representations of the external fields and the system-field coupling operators to be used interchangeably.
As a result of the joint system-field evolution described by the unitary operator $U(t)$ from (\ref{dU}), the output field $Y:= (Y_k)_{1\<k\<m}$ is given by
\begin{equation}
\label{Y}
    Y(t)
    =
    2
    \left[
    \begin{array}{c}
        \Re \fb(t)\\
        \Im \fb(t)
    \end{array}
    \right]
    =
 U(t)^{\dagger}(\cI_{\cH_0}\ox W(t))U(t)
\end{equation}
and satisfies the QSDE
\begin{equation}
\label{dY}
  \rd Y = 2Jh\rd t + \rd W.
\end{equation}
Here, the vectors $\fb:= (\fb_k)_{1\<k\<m/2}$ and $\fb^{\#}:= (\fb_k^{\dagger})_{1\<k\<m/2}$ consist of the corresponding output annihilation and creation operators:
$$
    \fb   = U^{\dagger}(\cI_{\cH_0}\ox \fa)U,
    \qquad
    \fb^{\#}   = U^{\dagger}(\cI_{\cH_0}\ox \fa^{\#})U.
$$
In view of (\ref{Wfa}), (\ref{hL}) and (\ref{dY}), the processes $\fb$ and $\fb^{\#}$ satisfy
the QSDEs
$$
    \rd \fb = L\rd t + \rd \fa,
    \qquad
    \rd \fb^{\#} = L^{\#}\rd t + \rd \fa^{\#}
$$
which are related to each other by conjugation.
The QSDEs (\ref{dX}) and (\ref{dY}) describe a particular yet important scenario of quantum stochastic dynamics with the identity scattering matrix \cite{HP_1984,P_1992}.
Endowed with additional features (for example, more general scattering matrices describing photon exchange between the fields),
such QSDEs are employed in a unified formalism  \cite{GJ_2009,JG_2010} for modelling feedback networks of quantum systems which  interact with each other and the external fields.

The unitary evolution in (\ref{uni}) and (\ref{Y}) preserves the commutativity between the system and output field variables in the sense that
\begin{equation}
\label{XY}
        [X(t),Y(s)^{\rT}]
     =
    0,
    \qquad
    t\> s\> 0
\end{equation}
(future system variables commute with the  past output variables).
However, the output fields $Y_1, \ldots, Y_m$ do not commute with each other since
\begin{equation}
\label{YY}
    [\rd Y, \rd Y^{\rT}]
    =
    [\rd W, \rd W^{\rT}] = 2iJ\rd t.
\end{equation}
In view of (\ref{WWst}) and (\ref{dY}), the process $Y$ inherits from $W$ the two-point commutator matrix:
\begin{equation}
\label{YYcom}
    [Y(s), Y(t)^{\rT}] = 2i\min(s,t)J,
    \qquad
    s,t\>0.
\end{equation}
Despite the system-output commutativity (\ref{XY}), the noncommutativity of the output fields in (\ref{YYcom})  makes them inaccessible to simultaneous measurement. However, they can be fed as an input to other open quantum systems without  conversion of the quantum operator-valued processes to classical real-valued signals. The resulting  coherent  field-mediated    interconnection can be used in fully quantum communication channels %\cite{?}
in the form of cascades of (possibly distant) quantum systems. Such cascade connections are also employed for the generation of certain classes of Gaussian quantum states (see, for example, \cite{MWPY_2014,VPJ_2017a}).
%  and will be discussed in Section~\ref{sec:cascade}.

%%%%%%%%%%%%%%%%%%%%%%%%%%%%%%%%%%%%%%%%%%%%%%%%%%%%%%%%%%%%%%%%%%%%%%%%%%%%%%%%%%%%%%%%%%%%%%%%%%%
\section{Evolution of the quasi-characteristic function}\label{sec:IDE}
%%%%%%%%%%%%%%%%%%%%%%%%%%%%%%%%%%%%%%%%%%%%%%%%%%%%%%%%%%%%%%%%%%%%%%%%%%%%%%%%%%%%%%%%%%%%%%%%%%%

Whereas the QSDE (\ref{dxi}) pertains to the Heisenberg picture  of quantum dynamics, its dual Schr\"{o}dinger picture version,  known as the master equation \cite{WM_2010},  describes the evolution of the reduced density operator. The latter is a quantum counterpart of the probability distribution for a classical Markov diffusion process,  which evolves according to the Fokker-Planck-Kolmogorov equation (FPKE) specified by the infinitesimal generator  of the process  \cite{KS_1991,S_2008}.  However, in contrast to the classical case, the master equation carries the ``burden of the Hilbert space'' in the sense that it involves operator-valued quantities rather than usual functions of real variables (such as a probability density function).

The Wigner-Moyal phase-space approach \cite{Hi_2010,M_1949} reconciles these two pictures by dealing with the reduced density operator indirectly through the moments of the system variables. These moments (and  their time evolution) can be recovered from the quasi-characteristic function \cite{CH_1971,H_2010} considered below. To this end, the system variables are assumed to  satisfy the Weyl
canonical commutation relations (CCRs)
\begin{equation}
\label{CCR}
    \cW_{u+v}(t) = \re^{i u^{\rT}\Theta v} \cW_u(t) \cW_v(t),
    \qquad
    u,v\in \mR^n,
\end{equation}
for any time $t\>0$.
These are represented in terms of the unitary Weyl operators \cite{F_1989}
\begin{equation}
\label{cW}
  \cW_u(t) := \re^{iu^{\rT} X(t)} = \cW_{-u}(t)^{\dagger}
\end{equation}
which inherit their time dependence from the system variables.
Here, $\Theta$ is a constant nonsingular real antisymmetric matrix of order $n$ which specifies the commutator matrix \begin{equation}
\label{XX}
    [X(t), X(t)^{\rT}]
    :=    ([X_j(t),X_k(t)])_{1\< j,k\< n}
     =
     2i \Theta
\end{equation}
as an infinitesimal form of the Weyl CCRs (\ref{CCR}). Note that (\ref{CCR}) is closely related to the Baker-Campbell-Hausdorff formula for the exponentials of more general operators which commute with their commutator (see, for example, \cite[pp. 128--129]{GZ_2004}).

\subsection*{Example 1}
Let the system variables consist of pairs of conjugate position and momentum operators \cite{M_1998} comprising the vectors $q:= (q_k)_{1\<k\<n/2}$ and $p:= -i\d_q = (-i\d_{q_k})_{1\<k\<n/2}$, with the reduced Planck constant being set to $\hslash = 1$, and $q_1, \ldots, q_{n/2}$ denoting the Cartesian coordinates in the position space $\mR^{n/2}$.  Then, without loss of generality, they can be assembled into the vector
\begin{equation}
\label{Xqp}
    X
    :=
    \left[
    \begin{array}{c}
        q\\
        p
    \end{array}
    \right]
\end{equation}
(where, as before,  the time arguments are omitted for brevity).
The corresponding CCR matrix $\Theta$ in (\ref{CCR}) and (\ref{XX}) takes the form
\begin{equation}
\label{Thetaqp}
    \Theta
    =
    \frac{1}{2}
    \left[
    \begin{array}{rc}
        0 & 1\\
        -1 & 0
    \end{array}
    \right]
    \ox I_{n/2}
\end{equation}
and, up to a factor of $\frac{1}{2}$,    coincides with the symplectic structure matrix in classical Hamiltonian systems \cite{A_1989}.
\hfill$\blacktriangle$

In what follows, $\bE \zeta = \Tr(\rho \zeta)$ denotes the quantum expectation over the system-field tensor-product density operator \begin{equation}
\label{rho}
    \rho := \varpi \ox \ups,
\end{equation}
where $\varpi$ is the initial quantum state of the system, and $\ups$ is the vacuum state \cite{P_1992} of the input fields. The averaging of the Weyl operator $\cW_u$, associated with the system variables by (\ref{cW}), leads to the complex-valued quasi-characteristic function (QCF)
\begin{equation}
\label{Phi}
    \Phi(t,u):= \bE \cW_u(t) = \overline{\Phi(t,-u)},
    \qquad
    t\>0,\
    u \in \mR^n.
\end{equation}
Here, the second equality  describes the Hermitian property of $\Phi(t,u)$ with respect to its spatial argument $u$ and follows from the second equality in (\ref{cW}). Also, the relation $\cW_0 = \cI_{\cH}$ implies that $\Phi(t,0)=1$. However, in view of the Weyl CCRs (\ref{CCR}), the Bochner-Khinchin positiveness criterion \cite{GS_2004}  for the characteristic functions of classical probability distributions is replaced with its quantum mechanical weighted version \cite{CH_1971,H_2010}: the complex Hermitian matrix $\big(\re^{iu_j^{\rT}\Theta u_k} \Phi(t,u_j-u_k)\big)_{1\<j,k\<\ell}$ is positive semi-definite for arbitrary points $u_1, \ldots, u_{\ell} \in \mR^n$ and any $\ell=1,2,3,\ldots$.
The spatial Fourier transform $\mF$ of (\ref{Phi}) yields a real-valued quasi-probability density function  (QPDF)
\begin{eqnarray}
\nonumber
    \mho(t,x)
    & := &
    \mF(\Phi(t,\cdot))(x)\\
\label{mho}
    & = &
    (2\pi)^{-n}\int_{\mR^n} \Phi(t,u)\re^{-iu^{\rT}x}\rd u,
    \qquad
    x \in \mR^n.
\end{eqnarray}
Although the function $\mho$ is not necessarily nonnegative everywhere \cite{H_1974} (since, as mentioned above, the QCF $\Phi$ itself does not have to be positive semi-definite), it satisfies the normalization condition $    \int_{\mR^n}\mho(t,x)\rd x = \Phi(t,0)=1$ and is a quantum analogue of classical PDFs. In particular, $\Phi$ and $\mho$  encode  the mixed moments $\bE(X_{k_1}\x \ldots \x X_{k_{\ell}})$ of the system variables  for any $1\< k_1,\ldots, k_{\ell}\< n$, provided $\Phi(t,u)$ is $\ell$ times continuously differentiable  with respect to $u\in \mR^n$.  For example, the first and second-order  moments of the system variables are expressed in terms of the QCF $\Phi$ and the QPDF $\mho$ similarly (modulo taking the real part) to the corresponding moments of classical random variables:
\begin{eqnarray}
\nonumber
        \bE X(t)
        & = &
    -i\d_u\Phi(t,0)\\
\label{EX}
    & = &
    \int_{\mR^n}
    \mho(t,x)x \rd x,\\
\nonumber
    \Re \bE(X(t)X(t)^{\rT})
    & = &
    -\d_u^2\Phi(t,0)\\
\label{EXX}
    & = &
    \int_{\mR^n}
    \mho(t,x) xx^{\rT} \rd x,
\end{eqnarray}
where $\d_u(\cdot)$ and $\d_u^2(\cdot)$ denote the gradient vector and the Hessian matrix with respect to $u\in \mR^n$.
Furthermore, the QCF $\Phi$ in (\ref{Phi})
can be used for evaluating the generalized moments
\begin{equation}
\label{EfX}
    \bE f(X(t)) = \int_{\mR^n}\sigma(u) \Phi(t,u)\rd u
\end{equation}
which involve nonlinear (but not necessarily polynomial) functions
\begin{equation}
\label{fX}
    f(X(t)):= \int_{\mR^n}\sigma(u) \cW_u(t)\rd u
\end{equation}
of the system variables  represented in the Weyl quantization form. Here, $\sigma: \mR^n\to \mC$ is a given function,  which specifies such a moment and can be a generalized function \cite{V_2002}. For example, (\ref{EX}) and (\ref{EXX}) can be obtained from (\ref{EfX}) and (\ref{fX}) by letting $\sigma$ be distributional partial derivatives  (of up to the second order) of the $n$-dimensional Dirac delta function $\delta$.

Since the QCF $\Phi$ in (\ref{Phi}) is the expectation of the Weyl operators, its time evolution satisfies the integro-differential equation (IDE)
\begin{equation}
\label{Phidot0}
    \d_t \Phi(t,u)
    =
    \frac{\bE \rd \cW_u(t) }{\rd t}
    =
    \bE \cL(\cW_u(t)).
\end{equation}
This relation is obtained by averaging the QSDE (\ref{dxi}),  which is applied to $\xi:=\cW_u$, with the martingale part $-i[\xi, h^{\rT}] \rd W$ not contributing to the right-hand side of (\ref{Phidot0}) (since the forward Ito increments of $W$ commute with adapted processes, as mentioned before,  and the input fields are in the vacuum state).

In general, the right-hand side of (\ref{Phidot0}) is not easily related to the QCF $\Phi$.
However, (\ref{Phidot0}) becomes an algebraically  closed equation for the QCF in the framework of the Weyl quantization model for the energy operators of the system.
More precisely, following \cite{V_2015c}, we assume that the system Hamiltonian $h_0$ and the system-field coupling operators $h_1, \ldots, h_m$ in (\ref{cL}) are obtained by the Weyl quantization  \cite{F_1989} of real-valued functions on $\mR^n$ with the Fourier transforms $H_k: \mR^n\to \mC$ as
\begin{eqnarray}
\nonumber
    h_k
    & := &
    \int_{\mR^n}
    H_k(u)\cW_u
    \rd u \\
\label{hk}
    & = &
    \int_{\mR^n}
    |H_k(u)|
    \cos
    (
        u^{\rT}X + \arg H_k(u)
    ) \rd u
    ,
    \qquad
    k = 0,1,\ldots, m,
\end{eqnarray}
where $\cW_u$ is the Weyl operator (\ref{cW}), and $\arg(\cdot)$ denotes the argument of a complex number. The second equality in (\ref{hk}) employs the Hermitian property of the functions $H_k$
 (that is, $H_k(-u) = \overline{H_k(u)}$ for all $u\in \mR^n$), which, in combination with the second equality in (\ref{cW}),  ensures that the operators $h_k$ are self-adjoint.
Accordingly, the vector $h$ of the system-field coupling operators $h_1, \ldots, h_m$ is related to the vector-valued map $H:= (H_k)_{1\< k\<m}: \mR^n \to \mC^m$ by
\begin{equation}
\label{h}
    h
    =
    \int_{\mR^n}
    H(u)
    \cW_u
    \rd u.
\end{equation}
If the function $H_k$ is absolutely integrable, then (\ref{hk}) can be understood as a Bochner integral  \cite{Y_1980} which yields a bounded operator $h_k$ whose operator norm admits an upper bound
$$
\|h_k\|\< \int_{\mR^n} |H_k(u)|\rd u<+\infty
$$
due to unitarity of the Weyl operators $\cW_u$ for all $u \in \mR^n$.
 Polynomial functions of the system variables can be obtained by letting the Fourier transforms $H_k$ in (\ref{hk}) be  the derivatives of the Dirac delta function (as mentioned before in regard to (\ref{EX})--(\ref{fX})).

The following theorem\footnote{whose proof, given here for completeness, is slightly different from the original one in \cite{V_2015c}} uses an integral operator $\fA$ which maps a function $\varphi: \mR^n \to \mC$ to the function $\fA(\varphi): \mR^n\to \mC$ given by
\begin{eqnarray}
\label{fA}
    \fA(\varphi)(u)
    & :=
    \int_{\mR^n}
    V(u,v)\varphi(u+v)
    \rd v,
    \qquad
    u \in \mR^n.
\end{eqnarray}
Here,
the kernel function $V: \mR^n \x \mR^n\to \mC$  is computed as
\begin{eqnarray}
\nonumber
  V(u,v)
  &:= &
    -2        \sin(u^{\rT}\Theta v)H_0(v)\\
\label{V}
     & &-2 \int_{\mR^n}
            \sin(u^{\rT}\Theta s)H(s)^{\rT}
            K(u+s,s-v)
        H(v-s) \rd s,
\end{eqnarray}
where $H_0$ and $H$  are the Fourier transforms from (\ref{hk}) and (\ref{h}), and the function
$K: \mR^n\x \mR^n\to \mR^{m\x m}$ is expressed as
\begin{eqnarray}
\nonumber
    K(u,v)
     & := &
         \Im
    \big(
    \re^{iu^{\rT}\Theta v}\Omega
    \big)\\
\label{K}
    & = &
    \sin(u^{\rT}\Theta v) I_m + \cos(u^{\rT}\Theta v) J
\end{eqnarray}
in terms of the CCR matrix $\Theta$ in (\ref{CCR}) and the quantum Ito matrix $\Omega$ from (\ref{Omega}). Note that (\ref{V}) depends linearly on $H_0$ and quadratically on $H$ in accordance with the way the Hamiltonian $h_0$ and the coupling operators $h_1, \ldots, h_m$ enter the GKSL generator $\cL$ in (\ref{cL}). Also, the integral in (\ref{V}) is organised as a weighted convolution of the function $H$ with itself.

%%%%%%%%%%%%%%%%%%%%%%%%%%%%%%%%%%%%%%%%%%%%%%%%%%%%%%%%%%%%%%%%%%%%%%%%%%%%%%%%%%%%%%%%%%%%%%%%%%%%%%%%%%%%%%%
\begin{thm}
\label{th:Phidot}\cite[Theorem~1]{V_2015c}
Suppose the energy operators $h_0,h_1, \ldots, h_m$ of the quantum stochastic system (\ref{dX}) have the Weyl quantization form (\ref{hk}). Then the IDE (\ref{Phidot0}) takes the form
\begin{equation}
\label{Phidot}
    \d_t \Phi(t,u)
    =
    \fA(\Phi(t,\cdot))(u),
\end{equation}
where the integral operator  $\fA$, described by  (\ref{fA})--(\ref{K}), acts over the spatial argument of the QCF $\Phi$.
\hfill$\square$
\end{thm}
%%%%%%%%%%%%%%%%%%%%%%%%%%%%%%%%%%%%%%%%%%%%%%%%%%%%%%%%%%%%%%%%%%%%%%%%%%%%%%%%%%%%%%%%%%%%%%%%%%%%%%%%%%%%%%%
\begin{proof}
By substituting the Weyl quantization of the system Hamiltonian $h_0$ from (\ref{hk}) into the GKSL generator $\cL$ in (\ref{cL}), it follows that
\begin{eqnarray}
\nonumber
    i[h_0,\cW_u]
    & = &
    i
    \int_{\mR^n}
    H_0(v)[\cW_v, \cW_u]
    \rd v\\
\label{cL0}
    & = &
    -2
    \int_{\mR^n}
    \sin(u^{\rT}\Theta v)
    H_0(v)
    \cW_{u+v}
    \rd v.
\end{eqnarray}
Here, use is made of the Lie algebraic property
\begin{eqnarray}
\nonumber
[\cW_u, \cW_v]
    & = &
    (\re^{-iu^{\rT}\Theta v}-\re^{iu^{\rT}\Theta v})
    \cW_{u+v}\\
\label{CCRcomm}
    & = &
    -2i
    \sin(u^{\rT}\Theta v)\cW_{u+v},
    \qquad
    u,v\in \mR^n,
\end{eqnarray}
of the Weyl operators, which follows from the CCRs (\ref{CCR}).
In combination with (\ref{CCRcomm}),  the Weyl quantization of the system-field coupling operators in (\ref{h}) leads to
\begin{eqnarray}
\nonumber
    [h^{\rT},\cW_u]\Omega h
    & = &
    \int_{\mR^n}
    H(s)^{\rT}
    [\cW_s, \cW_u]
    \rd s
    \Omega h\\
\nonumber
    & = &
    2i
    \int_{\mR^n}
    \sin(u^{\rT}\Theta s)
    H(s)^{\rT}
    \cW_{u+s}
    \rd s
    \Omega
    \int_{\mR^n}
    H(r)\cW_r\rd r\\
\nonumber
    & = &
    2i
    \int_{\mR^n\x \mR^n}
    \sin(u^{\rT}\Theta s)
    H(s)^{\rT}\Omega H(r)
    \cW_{u+s}\cW_r
    \rd s
    \rd r\\
\nonumber
    & = &
    2i
    \int_{\mR^n\x \mR^n}
    \sin(u^{\rT}\Theta s)
    \re^{-i(u+s)^{\rT}\Theta r}
    H(s)^{\rT}\Omega H(r)
    \cW_{u+s+r}
    \rd s
    \rd r\\
\label{cL1}
    & = &
    \int_{\mR^n}
    V_1(u,v)
    \cW_{u+v}
    \rd v,
\end{eqnarray}
where (\ref{CCR}) is used again. Here, the kernel function $V_1: \mR^n\x \mR^n\to \mC$ is organised as a weighted convolution of the function $H$ with itself:
\begin{equation}
\label{V1}
  V_1(u,v)
  =
    2i
    \int_{\mR^n}
    \sin(u^{\rT}\Theta s)
    \re^{i(u+s)^{\rT}\Theta (s-v)}
    H(s)^{\rT}\Omega H(v-s)
    \rd s.
\end{equation}
By a similar reasoning,
\begin{eqnarray}
\nonumber
    h^{\rT}\Omega [\cW_u,h]
    & = &
    h^{\rT}\Omega
    \int_{\mR^n}
    H(s)
    [\cW_u, \cW_s]
    \rd s\\
\nonumber
    & = &
    -2i
    \int_{\mR^n}
    H(r)^{\rT}
    \cW_r
    \rd r
    \Omega
    \int_{\mR^n}
    \sin(u^{\rT}\Theta s)
    H(s)
    \cW_{u+s}
    \rd s\\
\nonumber
    & = &
    -2i
    \int_{\mR^n\x \mR^n}
    \sin(u^{\rT}\Theta s)
    H(r)^{\rT}\Omega H(s)
    \cW_r\cW_{u+s}
    \rd r
    \rd s\\
\nonumber
    & = &
    -2i
    \int_{\mR^n\x \mR^n}
    \sin(u^{\rT}\Theta s)
    \re^{-ir^{\rT}\Theta (u+s)}
    H(r)^{\rT}\Omega H(s)
    \cW_{u+s+r}
    \rd r
    \rd s\\
\label{cL2}
    & = &
    \int_{\mR^n}
    V_2(u,v)
    \cW_{u+v}
    \rd v.
\end{eqnarray}
Here, in view of the antisymmetry of the CCR matrix $\Theta$ and the Hermitian property of the quantum Ito matrix $\Omega$,
the kernel function $V_2: \mR^n\x \mR^n\to \mC$  is computed as
\begin{eqnarray}
\nonumber
  V_2(u,v)
  & = &
    -2i
    \int_{\mR^n}
    \sin(u^{\rT}\Theta s)
    \re^{i(s-v)^{\rT}\Theta (u+s)}
    H(v-s)^{\rT}\Omega H(s)
    \rd s\\
\nonumber
  & = &
    -2i
    \int_{\mR^n}
    \sin(u^{\rT}\Theta s)
    \re^{-i(u+s)^{\rT}\Theta (s-v)}
    H(s)^{\rT}\Omega^{\rT} H(v-s)
    \rd s\\
\label{V2}
  & = &
    -2i
    \int_{\mR^n}
    \sin(u^{\rT}\Theta s)
    \re^{-i(u+s)^{\rT}\Theta (s-v)}
    H(s)^{\rT}\overline{\Omega} H(v-s)
    \rd s.
\end{eqnarray}
From (\ref{V1}) and (\ref{V2}), it follows that
\begin{eqnarray}
\nonumber
    \frac{1}{2}
    (V_1(u,v) + V_2(u,v))
    & = &
    i
    \int_{\mR^n}
    \sin(u^{\rT}\Theta s)
    H(s)^{\rT}
    \big(
        \re^{i(u+s)^{\rT}\Theta (s-v)}\Omega
        -
        \re^{-i(u+s)^{\rT}\Theta (s-v)}
    \overline{\Omega}
    \big) H(v-s)
    \rd s\\
\nonumber
    & = &
    -2
    \int_{\mR^n}
    \sin(u^{\rT}\Theta s)
    H(s)^{\rT}
    \Im
    \big(
        \re^{i(u+s)^{\rT}\Theta (s-v)}\Omega
    \big)
    H(v-s)
    \rd s    \\
\label{V12}
    & = &
    -2
    \int_{\mR^n}
    \sin(u^{\rT}\Theta s)
    H(s)^{\rT}
    K(u+s,s-v)
    H(v-s)
    \rd s,
\end{eqnarray}
where use is made of the function $K$ from (\ref{K}). A combination of (\ref{cL0}), (\ref{cL1}) and (\ref{cL2})
 allows the superoperator $\cL$ in  (\ref{cL}) to be evaluated at the Weyl operator as
\begin{equation}
\label{cLcW}
    \cL(\cW_u)
    =
    \int_{\mR^n}
    V(u,v)\cW_{u+v}
    \rd v,
\end{equation}
 where
 $
    V(u,v) =
    -2
    \sin(u^{\rT}\Theta v)
    H_0(v)
    +
    \frac{1}{2}
    (V_1(u,v) + V_2(u,v))
 $
 is given by (\ref{V}) in view of (\ref{V12}).   The averaging of (\ref{cLcW}) represents (\ref{Phidot0}) in the form $    \d_t \Phi(t,u)
    =
    \int_{\mR^n}
    V(u,v)
    \Phi(t,u+v)
    \rd v
$,
which establishes  (\ref{Phidot}), with the integral operator $\fA$ in (\ref{fA}).
\end{proof}
%%%%%%%%%%%%%%%%%%%%%%%%%%%%%%%%%%%%%%%%%%%%%%%%%%%%%%%%%%%%%%%%%%%%%%%%%%%%%%%%%%%%%%%%%%%%%%%%%%%%%%%%%%%%%%%

If the system and fields are uncoupled,  the IDE (\ref{Phidot}) reduces to the Moyal equation for isolated quantum systems \cite{M_1949}:
\begin{equation}
\label{Moy}
    \d_t\Phi(t,u)
    =
    -2
    \int_{\mR^n}
    \sin(u^{\rT}\Theta v)
    H_0(v)
    \Phi(t,u+v)
    \rd v,
\end{equation}
which follows from (\ref{fA}) by letting $H=0$ in (\ref{V}). Note that (\ref{Phidot}) and its special case (\ref{Moy}) are phase-space representations of the master equation for the reduced density operator in the spatial frequency domain.
The corresponding IDE for the QPDF $\mho$ in (\ref{mho}) is obtained from (\ref{Phidot}) by a unitary similarity transformation of the operator $\fA$ (due to unitarity of the Fourier transform $\mF$ up to a factor of $(2\pi)^{n/2}$):
\begin{eqnarray}
\nonumber
    \d_t \mho(t,x)
    & = &
    \int_{\mR^n}
    \Pi(x,y)
    \mho(t,y)
    \rd y\\
\label{mhodot}
    & = &
    \fF(\mho(t,\cdot))(x),
    \qquad
    \fF:= \mF \fA\mF^{-1},
\end{eqnarray}
where $\mF^{-1}$ is the inverse Fourier transform over the spatial variables in $\mR^n$. The kernel function $\Pi$ of the integral operator $\fF$ is related to $V$ in (\ref{fA}) by
\begin{equation}
\label{Pi}
    \Pi(x,y)
    =
    (2\pi)^{-n}
    \int_{\mR^n\x \mR^n}
    \re^{i(u^{\rT}(y-x)+v^{\rT}y)}
    V(u,v)
    \rd u\rd v.
\end{equation}

%%%%%%%%%%%%%%%%%%%%%%%%%%%%%%%%%%%%%%%%%%%%%%%%%%%%%%%%%%%%%%%%%%%%%%%%%%%%%%%%%%%%%%%%%%%%%%%%%%%
\section{Linear-Gaussian dynamics of open quantum harmonic oscillators}\label{sec:gauss}
%%%%%%%%%%%%%%%%%%%%%%%%%%%%%%%%%%%%%%%%%%%%%%%%%%%%%%%%%%%%%%%%%%%%%%%%%%%%%%%%%%%%%%%%%%%%%%%%%%%

We will now consider the case
when the Hamiltonian $h_0$ is a quadratic function and the coupling operators $h_1, \ldots, h_m$ are linear functions  of the system variables described by
\begin{eqnarray}
\label{hhOQHO1}
  h_0
  & := &
  \frac{1}{2} X^{\rT} R X,\\
\label{hhOQHO2}
  h
  & := &
  MX,
\end{eqnarray}
where $R$ is a real symmetric matrix of order $n$, and $M \in \mR^{m\x n}$.
 %with rows $M_1, \ldots, M_m$, so that $h_j = M_j X$.
 The matrices $R$ and $M$ will be referred to as the \emph{energy} and \emph{coupling matrices}, respectively. The corresponding Fourier transforms $H_0$ and $H$ in (\ref{hk}) and (\ref{h}) are given by
\begin{eqnarray}
\label{HH1}
    H_0(u)
      & = &
    -\frac{1}{2}
    \bra
        R,
        \delta''(u)
    \ket,\\
\label{HH2}
    H(u)
    &  = &
    iM\delta'(u),
\end{eqnarray}
where $\delta'$ and $\delta''$ are the distributional gradient vector and Hessian matrix of the $n$-dimensional Dirac delta function   $\delta$, and $\bra N, L\ket:= \Tr(N^*L)$ denotes the Frobenius inner product  \cite{HJ_2007} of real or complex matrices.

%%%%%%%%%%%%%%%%%%%%%%%%%%%%%%%%%%%%%%%%%%%%%%%%%%%%%%%%%%%%%%%%%%%%%%%%%%%%%%%%%%%%%%%%%%%%%%%%%%%
\subsection*{Example 2}

Suppose the system variables are the quantum mechanical positions and momenta as in (\ref{Xqp}) of Example~1, and the system Hamiltonian
consists of the quadratic potential and kinetic energy parts as
\begin{equation}
\label{pot}
    h_0
    =
    \frac{1}{2}
    (q^{\rT} \sK q + p^{\rT} \sM^{-1} p),
\end{equation}
with a stiffness matrix $\sK$ and a positive definite mass matrix $\sM$. Then  the representation (\ref{hhOQHO1}) holds with a block-diagonal energy matrix
\begin{equation}
\label{Rqp}
    R =
    \left[
    \begin{array}{cl}
        \sK & 0 \\
        0 &  \sM^{-1}
    \end{array}
    \right].
\end{equation}
\hfill$\blacktriangle$
%%%%%%%%%%%%%%%%%%%%%%%%%%%%%%%%%%%%%%%%%%%%%%%%%%%%%%%%%%%%%%%%%%%%%%%%%%%%%%%%%%%%%%%%%%%%%%%%%%%

Regardless of a particular structure of the system variables and the energy matrix, considered  in the above example, the combination of the CCRs (\ref{XX}) and the linear-quadratic energetics in (\ref{hhOQHO1}) and (\ref{hhOQHO2}) makes the QSDEs (\ref{dX}) and (\ref{dY}) linear:
\begin{eqnarray}
\label{dXlin}
    \rd X
    & = &
    AX \rd t + B \rd W,\\
\label{dYlin}
    \rd Y
    & = &
    CX\rd t + \rd W,
\end{eqnarray}
where $A\in \mR^{n\x n}$, $B\in \mR^{n\x m}$ and $C \in \mR^{m\x n}$ are constant matrices of coefficients:
\begin{eqnarray}
\label{A}
    A
    & := &
    2\Theta (R + M^{\rT}JM),\\% = 2\Theta R - \frac{1}{2}BJB^{\rT}\Theta^{-1},
\label{B}
    B
    &:= &
    2\Theta M^{\rT},\\
\label{C}
    C
    & := &
    2JM.
\end{eqnarray}
The special dependence of $A$, $B$, $C$ on the  energy and coupling matrices $R$ and $M$ from (\ref{hhOQHO1}) and (\ref{hhOQHO2}), the symmetry of $R$ and the antisymmetry of the matrices $\Theta$ and $J$,  lead to the physical realizability (PR) conditions \cite{JNP_2008,SP_2012}:
\begin{eqnarray}
\label{APR}
    A \Theta + \Theta A^{\rT} + BJB^{\rT}
    & = &
    0,\\
\label{BCPR}
    \Theta C^{\rT} + BJ
    & = &
    0,
\end{eqnarray}
the first of which is closely related to the preservation in time of the CCRs (\ref{XX}), while (\ref{BCPR}) pertains to (\ref{XY}).
Such quantum stochastic systems are referred to as open quantum harmonic oscillators (OQHOs).
Their QCF and QPDF dynamics can be obtained from Theorem~\ref{th:Phidot} by using the Weyl quantization and the generalized functions in the spatial frequency domain as the following lemma shows for completeness.
%
%%%%%%%%%%%%%%%%%%%%%%%%%%%%%%%%%%%%%%%%%%%%%%%%%%%%%%%%%%%%%%%%%%%%%%%%%%%%%%%%%%%%%%%%%%%%%%%%%%%%%%%%%
\begin{lem}
Suppose the initial system variables of the OQHO, described by (\ref{dXlin})--(\ref{C}) and driven by vacuum input fields, have finite second moments (that is, $\bE(X(0)^{\rT}X(0)) < +\infty$). Then the QCF $\Phi(t,\cdot)$ in (\ref{Phi}) remains twice continuously differentiable with respect to its spatial variables for any time $t\> 0$ and satisfies the first-order linear PDE
\begin{equation}
\label{PhiPDE}
    \d_t \Phi(t,u)
    =
    u^{\rT} A \d_u \Phi(t,u)
        -\frac{1}{2}
        |B^{\rT}u|^2
    \Phi(t,u).
\end{equation}
Also, the QPDF $\mho$ in (\ref{mho}) is governed by the parabolic PDE
\begin{equation}
\label{mhoPDE}
    \d_t\mho(t,x)
    =
    -\div(
        \mho(t,x)Ax
    )
    +
    \frac{1}{2}
    \div^2(
        \mho(t,x) BB^{\rT}
    ),
 \end{equation}
 where $\div(\cdot)$ is the divergence operator over the spatial variables (which acts  on matrix-valued functions $f:= (f_{jk})_{1\< j\< m,1\< k\< n}$ on $\mR^n$ with $m>1$ in a row-wise fashion as $\div f := \big(\sum_{k=1}^n \nabla_k f_{jk}\big)_{1\< j\< m}$, where $\nabla_k$ is the partial derivative with respect to the $k$th Cartesian coordinate in $\mR^n$).
\hfill$\square$
\end{lem}
%%%%%%%%%%%%%%%%%%%%%%%%%%%%%%%%%%%%%%%%%%%%%%%%%%%%%%%%%%%%%%%%%%%%%%%%%%%%%%%%%%%%%%%%%%%%%%%%%%%%%%%%%
\begin{proof}
Substitution of (\ref{HH1}) and (\ref{HH2}) into (\ref{V}) allows the action of the integral operator $\fA$ in (\ref{fA}) on a twice continuously differentiable function $\varphi$ to be represented in the form
\begin{eqnarray}
\nonumber
  \fA(\varphi)(u)
  &:= &
    \int_{\mR^n}
    \sin(u^{\rT}\Theta v)
    \bra
        R,
        \delta''(v)
    \ket
    \varphi(u+v)
    \rd v \\
\nonumber
    & & +2 \int_{\mR^n}
            \sin(u^{\rT}\Theta s)
            \delta'(s)^{\rT}M^{\rT}
            \Big(\int_{\mR^n}
            K(u+s,s-v)
            M \delta'(v-s)
            \varphi(u+v)
            \rd v
            \Big)
            \rd s\\
\nonumber
  &= &
    \Bra
        R,
        \d_v^2(\sin(u^{\rT}\Theta v)
    \varphi(u+v))\big|_{v = 0}
    \Ket\\
\nonumber
   & & -2 \int_{\mR^n}
            \sin(u^{\rT}\Theta s)
            \delta'(s)^{\rT}M^{\rT}
            \div_v\big(
            \varphi(u+v)
            K(u+s,s-v)
            M
            \big)\big|_{v=s}
            \rd s\\
\nonumber
  &= &
    2
    \varphi'(u)^{\rT}
    R
    \d_v\sin(u^{\rT}\Theta v)\big|_{v=0}\\
\nonumber
   & & +2
     \div_s
     \big(
            \sin(u^{\rT}\Theta s)
            M^{\rT}
            \div_v\big(
            \varphi(u+v)
            K(u+s,s-v)
            M
            \big)\big|_{v=s}
     \big)
     \big|_{s=0}\\
\nonumber
 & = &
    -2
    \varphi'(u)^{\rT}
    R
    \Theta u \\
\nonumber
  &  &
    +2
     \div_s
     \big(
            \sin(u^{\rT}\Theta s)
            M^{\rT}
            \div_v\big(
            \varphi(u+v)
            K(u+s,s-v)
            M
            \big)\big|_{v=s}
     \big)
     \big|_{s=0}\\
\nonumber
 & = &
    2
    u^{\rT}
    \Theta R
    \varphi'(u)\\
\nonumber
  &  &
    +2
     \div_s
     \big(
            \sin(u^{\rT}\Theta s)
            M^{\rT}
            (
                JM\varphi'(u+s)
                +
                \varphi(u+s) M\Theta (u+s)
            )
     \big)
     \big|_{s=0}\\
\nonumber
  &= &
    2
    u^{\rT}
    \Theta R
    \varphi'(u)\\
\nonumber
   & &
    +2
     \div_s
     \big(
            \sin(u^{\rT}\Theta s)
            M^{\rT}
            (
                JM\varphi'(u+s)
                +
                \varphi(u+s) M\Theta (u+s)
            )
     \big)
     \big|_{s=0}\\
\nonumber
 & = &
    2
    u^{\rT}
    \Theta R
    \varphi'(u)
     +
     2
     u^{\rT}\Theta
     \big(
            M^{\rT}
            (
                JM\varphi'(u)
                +
                \varphi(u) M\Theta u
            )
     \big)\\
\label{fAphi}
  &  = &
        u^{\rT} A \d_u \varphi(u)
        -\frac{1}{2}
        |B^{\rT}u|^2
    \varphi(u),
\end{eqnarray}
where $\div_v(\cdot)$ and $\div_s(\cdot)$ are the divergence operators with respect to the vectors $v,s \in \mR^n$, and use is made of the matrices $A$ and $B$ from (\ref{A}) and (\ref{B}) together with the symmetry of $R$ and antisymmetry of $\Theta$.
Here, we have also used the relation
\begin{equation}
\label{fdeldash}
    f \delta'
    =
    f(0)\delta'-\div f(0)\delta
\end{equation}
for a continuously differentiable function $f:=(f_{jk})_{1\< j\< m, 1\< k\< n}:\mR^n\to \mC^{m\x n}$ and the distributional gradient of the $n$-dimensional Dirac delta function (see, for example, \cite{V_2002}).
Now, if the initial system variables of the OQHO have finite second moments, then the linear QSDE (\ref{dXlin}) propagates this mean square integrability to all subsequent times as
\begin{eqnarray}
\nonumber
    \bE(X(t)^{\rT}X(t))
    & = &
    \bE(X(0)^{\rT} \re^{tA^{\rT}} \re^{tA} X(0))
    +
    \int_0^t
    \bra
        B^{\rT}\re^{(t-s)A^{\rT}} \re^{(t-s)A} B,
        \Omega
    \ket
    \rd s
    \\
\label{Xint}
    & \< &
    \|\re^{tA}\|^2
    \bE(X(0)^{\rT}X(0))
    +
    m
    \int_0^t
    \|\re^{sA} B\|^2\rd s,
\end{eqnarray}
where $\|\cdot\|$ is the operator norm of a matrix, and the factor $m = \Tr \Omega$ comes from the quantum Ito matrix in (\ref{Omega}). This follows from the relation
\begin{equation}
\label{Xt}
    X(t)
    =
    \re^{t A} X(0)
    +
    \int_0^t
    \re^{(t-s)A} B
    \rd W(s),
\end{equation}
the product structure (\ref{rho}) of the system-field density operator and the input fields being in the vacuum state.
The property $\bE(X(t)^{\rT}X(t))<+\infty$ in (\ref{Xint}) secures the twice continuous differentiability of the QCF $\Phi(t,\cdot)$  with respect to its spatial argument for any $t\>0$. Therefore, in view of (\ref{fAphi}), the IDE (\ref{Phidot}) takes the form of the PDE (\ref{PhiPDE}). Accordingly,
the IDE (\ref{mhodot}) for the QPDF $\mho$ with the kernel function (\ref{Pi}) can be reduced to the PDE (\ref{mhoPDE})   by applying the spatial Fourier transform  to  both sides of (\ref{PhiPDE}).
\end{proof}
%%%%%%%%%%%%%%%%%%%%%%%%%%%%%%%%%%%%%%%%%%%%%%%%%%%%%%%%%%%%%%%%%%%%%%%%%%%%%%%%%%%%%%%%%%%%%%%%%%%%%%%%%
%

Note that  (\ref{mhoPDE}) coincides with the FPKE for a classical Markov diffusion process $\xi$ in $\mR^n$ governed by the SDE $\rd \xi = A\xi \rd t + B\rd \omega$ (where $\omega$ is the standard Wiener process in $\mR^m$) with
the infinitesimal generator $\fF^{\dagger}$ which acts
on a smooth function $\varphi: \mR^n \to \mR$ (with the gradient $\varphi'$ and the Hessian matrix $\varphi''$) as
\begin{equation}
\label{gen}
    \fF^{\dagger}(\varphi)(x)
    :=
    x^{\rT}A^{\rT}\varphi'(x)
    +
    \frac{1}{2}
    \bra
        BB^{\rT},
        \varphi''(x)
    \ket
\end{equation}
 (see, for example, \cite{KS_1991,S_2008}).
The PDEs (\ref{PhiPDE}) and (\ref{mhoPDE}) show that the OQHO inherits the preservation of the Gaussian nature of quantum states from classical Markov processes governed by linear SDEs. The QCF of a Gaussian quantum state \cite{KRP_2010} is parameterized by the mean vector $\mu\in \mR^n$ and the real part $\Sigma$ of the quantum covariance matrix as
\begin{equation}
\label{Phigauss}
    \Phi_{\mu,\Sigma}(u) = \re^{i\mu^{\rT} u - \frac{1}{2}\|u\|_{\Sigma}^2},
    \qquad
    u \in \mR^n,
\end{equation}
where $\|u\|_{\Sigma}:= |\sqrt{\Sigma}u|$. The Gaussian QCF (\ref{Phigauss}) is identical to its classical counterpart except that $\Sigma$ satisfies the generalized form $\Sigma + i\Theta\succcurlyeq 0 $ of the Heisenberg uncertainty principle \cite{H_2001}, which is a stronger property than the positive semi-definiteness of $\Sigma$.
The corresponding Gaussian QPDF  is given by
 \begin{equation}
\label{mhogauss}
    \mho_{\mu,\Sigma}(x)
    :=
    \frac{(2\pi)^{-n/2}}{\sqrt{\det \Sigma}}\re^{-\frac{1}{2}\|x-\mu\|_{\Sigma^{-1}}^2},
    \qquad
    x\in \mR^n,
\end{equation}
provided $\Sigma\succ 0$ (the latter condition also makes $\Phi_{\mu,\Sigma}$ square integrable).
Therefore, if the system is initialized at a Gaussian quantum state, then the state remains Gaussian over the course of time,  with the parameters of the QCF (\ref{Phigauss}) satisfying the ODEs
\begin{eqnarray}
\label{mudot}
    \dot{\mu}
    & = &
    A\mu,\\
\label{Sigmadot}
    \dot{\Sigma}
    & = &
    A\Sigma +\Sigma A^{\rT} + BB^{\rT},
\end{eqnarray}
where the matrices $A$ and $B$ are given by (\ref{A}) and (\ref{B}). Furthermore, in this case, the system variables and the output fields, described by (\ref{Y})--(\ref{YYcom}), are in an augmented Gaussian state at every moment of time.
Furthermore, if $A$ is Hurwitz, then the system has a unique invariant quantum state. This state  is Gaussian with the mean vector $\mu=0$ (in view of (\ref{mudot})) and the real part of the quantum covariance matrix being the unique steady-state solution of the Lyapunov ODE in (\ref{Sigmadot}):
\begin{equation}
\label{Siginv}
    P
    :=
    \int_0^{+\infty}
    \re^{tA}
    BB^{\rT}
    \re^{tA^{\rT}}
    \rd t,
\end{equation}
which is the infinite-horizon Gramian of the matrix pair $(A,B)$, satisfying the algebraic Lyapunov equation (ALE)
\begin{equation}
\label{PALE}
    A P + PA^{\rT} + BB^{\rT} = 0.
\end{equation}
The property that the matrix $P$ is physically meaningful (that is, satisfies the Heisenberg uncertainty relation) can be verified directly  by noting that
\begin{equation}
\label{Ppos}
    P+i\Theta
    =
    \int_0^{+\infty}\re^{tA}B\Omega B^{\rT}\re^{tA^{\rT}}\rd t
    \succcurlyeq 0
\end{equation}
in view of (\ref{APR}), (\ref{Siginv}) and the positive semi-definiteness of the quantum Ito matrix $\Omega$.
The convergence to the invariant Gaussian quantum state also holds in the sense of weak convergence of probability measures \cite{B_1968} (appropriately modified for the quantum mechanical case \cite{CH_1971}) irrespective of whether the initial state is Gaussian. The invariant Gaussian QCF and QPDF
\begin{eqnarray}
\label{Phi*}
  \Phi_*(u)
  & := &
  \Phi_{0,P}(u)
  =
  \re^{-\frac{1}{2}\|u\|_P^2}, \\
\label{mho*}
  \mho_*(x)
  & := &
  \mho_{0,P}(x)
  =
  \frac{(2\pi)^{-n/2}}{\sqrt{\det P}}\re^{-\frac{1}{2}\|x\|_{P^{-1}}^2}
\end{eqnarray}
(with (\ref{mho*}) assuming that $(A,B)$ is controllable which ensures that $P\succ 0$)
are the fixed points of the corresponding semi-groups associated with the IDEs (\ref{Phidot}) and (\ref{mhodot}) in the sense that
\begin{equation}
\label{fix}
    \re^{t \fA}(\Phi_*) = \Phi_*,
    \qquad
    \re^{t \fF}(\mho_*) = \mho_*
\end{equation}
for all $t\>0$. In the case of OQHOs being considered, the operator exponential $\re^{t \fA}$ in (\ref{fix}) can be computed by solving the first-order PDE (\ref{PhiPDE}) with an initial condition $\varphi: \mR^n\to \mC$
through the method of characteristics \cite{E_1998,V_1971} as
\begin{eqnarray}
\nonumber
    \re^{t\fA} (\varphi)(u)
    & = &
    \varphi\big(\re^{tA^{\rT}}u\big)
    \Phi_{0, \Sigma(t)}(u)\\
\label{Phisol}
    & = &
    \varphi\big(\re^{tA^{\rT}}u\big)
    \re^{-\frac{1}{2}\|u\|_{\Sigma(t)}^2},
    \qquad
    t\>0,\
    u \in \mR^n.
\end{eqnarray}
Here, use is made of (\ref{Phigauss}) together with the finite-horizon controllability Gramian of the pair $(A,B)$ satisfying the ODE (\ref{Sigmadot}) with zero initial condition $\Sigma(0) = 0$:
\begin{equation}
\label{Sigma}
    \Sigma(t)
    =
    \int_0^t
    \re^{sA}
    BB^{\rT}
    \re^{sA^{\rT}}
    \rd s,
\end{equation}
with $\Tr \Sigma(t)$ being bounded by the second term on the right-hand side of (\ref{Xint}). The Gramians $\Sigma(t)$ and $P= \lim_{t\to +\infty} \Sigma(t)$ in (\ref{Siginv}) satisfy the identity
\begin{equation}
\label{PSig}
    \re^{tA} P \re^{tA^{\rT}} + \Sigma(t) = P,
    \qquad
    t\> 0.
\end{equation}
The operator exponential $\re^{t\fA}$ in (\ref{Phisol}) can also be obtained directly from (\ref{Xt}) (regardless of whether the system variables have finite second moments) as
\begin{eqnarray*}
    \Phi(t,u)
    & = &
    \bE
    \re^{
    iu^{\rT}
    (\re^{t A} X(0)
    +
    \int_0^t
    \re^{(t-s)A} B
    \rd W(s))}\\
    & = &
    \bE
    \re^{
    iu^{\rT}\re^{t A} X(0)}
    \bE\re^{iu^{\rT}
    \int_0^t
    \re^{(t-s)A} B
    \rd W(s)}\\
    & = &
    \Phi(0,\re^{tA^{\rT}}u)
    \re^{-\frac{1}{2}\int_0^t |B^{\rT} \re^{sA^{\rT}}u|^2 \rd s}\\
    & = &
    \Phi(0,\re^{tA^{\rT}}u)
    \re^{-\frac{1}{2}\|u\|_{\Sigma(t)}^2}.
\end{eqnarray*}
Here, use is made of the commutativity between the input fields (more precisely, their increments) and the initial system variables, the
product structure of the system-field state (\ref{rho}) and the quasi-characteristic functional of the vacuum fields:
$$
    \bE \re^{i\int_0^t f(s)^{\rT}\rd W(s)}
    =
    \re^{-\frac{1}{2}\int_0^t |f(s)|^2 \rd s}
$$
for any $t\>0$ and any locally square integrable function $f:\mR_+\to \mR^m$.
Accordingly, $\re^{t\fF}$ is an integral operator whose kernel function is the fundamental solution \cite{E_1998,V_1971} of the FPKE (\ref{mhoPDE}), provided $(A,B)$ is controllable, which is equivalent to $\Sigma(t)\succ 0$ for all $t>0$.  This operator maps an initial condition $\psi: \mR^n\to \mR$ of the FPKE to the function
\begin{equation}
\label{trans}
    \re^{t\fF}(\psi)(x)
    =
    \int_{\mR^n}
    \mho_{\re^{tA} y, \Sigma(t)}(x)
    \psi(y)
    \rd y,
    \qquad
    t>0,\
    x \in \mR^n,
\end{equation}
where use is made of the Gaussian PDF (\ref{mhogauss}).  The function $y\mapsto \mho_{\re^{tA} y, \Sigma(t)}(x)$ is the Markov transition kernel of the classical diffusion process with the  infinitesimal generator (\ref{gen}).

Therefore, any matrix pair $(A,B)$ satisfying (\ref{Siginv}) and  the PR condition (\ref{APR}), with $A$ Hurwitz, can be used in order to generate a zero mean Gaussian state with a given quantum covariance matrix  $P+i\Theta$ in (\ref{Ppos}) as an invariant state of the OQHO governed by (\ref{dXlin}). Note that inaccuracies in the energy and coupling matrices $R$ and $M$ of such a system will
only perturb the matrix $P$ for this invariant state without distorting its Gaussian nature. Their effects are considered, for example, in \cite{VPJ_2017a}. However, the Hamiltonian and coupling operators can, in general, be subject to a broader class of perturbations which bring the quantum system beyond the linear-quadratic energetics (\ref{hhOQHO1}) and (\ref{hhOQHO2}).

%%%%%%%%%%%%%%%%%%%%%%%%%%%%%%%%%%%%%%%%%%%%%%%%%%%%%%%%%%%%%%%%%%%%%%%%%%%%%%%%%%%%%%%%%%%%%%%%%%%
\section{Weyl variations of the energy operators}\label{sec:Weylvar}
%%%%%%%%%%%%%%%%%%%%%%%%%%%%%%%%%%%%%%%%%%%%%%%%%%%%%%%%%%%%%%%%%%%%%%%%%%%%%%%%%%%%%%%%%%%%%%%%%%%

Now, consider a class of quantum stochastic systems, which are similar to those in \cite{SVP_2014,V_2015c,V_2017} and  result from perturbing the linear-quadratic energy operators (\ref{hhOQHO1}) and (\ref{hhOQHO2}) of the nominal  OQHO of the previous section subject to the following variations:
\begin{eqnarray}
\label{h0eps}
  h_0^{\eps}
  & := &
  \frac{1}{2} X^{\rT} R X +
  \eps\wt{h_0},\\
\label{heps}
  h^{\eps}
  & := &
  MX + \eps\wt{h}.
\end{eqnarray}
Here, $\eps$ is a small real-valued parameter, and $\wt{h}_0$ and $\wt{h}:= (\wt{h}_j)_{1\< j\< m}$ describe  nonquadratic and nonlinear parts of the perturbed energy operators  $h_0^{\eps}, h_1^{\eps}, \ldots, h_m^{\eps}$,  which are assumed to be represented in the Weyl quantization form  (see Section~\ref{sec:IDE}) as
\begin{eqnarray}
\nonumber
  \wt{h}_0
  & := &
  \int_{\mR^d} \Psi(v) \cW_{S^{\rT}v} \rd v\\
\label{h0tilde}
  & = &
  \int_{\mR^d} |\Psi(v)|\cos(v^{\rT}SX + \arg \Psi(v)) \rd v,\\
\nonumber
  \wt{h}
  & := &
  \int_{\mR^d} \Ups(v) \cW_{S^{\rT}v} \rd v\\
\label{htilde}
  & = &
  \Big(\int_{\mR^d} |\Ups_k(v)| \cos(v^{\rT}SX + \arg \Ups_k(v))\rd v\Big)_{1\< k\< m}.
\end{eqnarray}
These \emph{Weyl variations} depend on $d\<n $ system variables
 comprising a subvector $S X$ (where $S\in \{0,1\}^{d\x n}$ is a suarray of a permutation matrix of order $n$) and are determined by the Fourier transforms $\Psi: \mR^d \to \mC$ and
$\Ups:= (\Ups_k)_{1\< k\< m}: \mR^d \to \mC^m$ of functions on $\mR^d$ with values in $\mR$ and $\mR^m$, respectively. In addition to being Hermitian, both functions $\Psi$ and $\Ups$ are assumed to be absolutely integrable, so that (\ref{h0tilde}) and (\ref{htilde}) are Bochner integrals which produce self-adjoint bounded operators on the system-field Hilbert space $\cH$.
The corresponding Fourier transforms in (\ref{hk}) and (\ref{h}) are given by
\begin{eqnarray}
\label{H0eps}
    H_0^{\eps}(u)
     & = &
    -\frac{1}{2}
    \bra
        R, \delta''(u)
    \ket + \eps\wt{H}_0(u),\\
\label{Heps}
    H^{\eps}(u)
    &= &
    iM\delta'(u) + \eps \wt{H}(u),
\end{eqnarray}
where
\begin{eqnarray}
\label{H0t}
    \wt{H}_0(u)
    & := &
    \int_{\mR^d}\Psi(v)\delta(u-S^{\rT}v)\rd v,\\
\label{Ht}
    \wt{H}(u)
    & := &
    \int_{\mR^d}\Ups(v)\delta(u-S^{\rT}v)\rd v.
\end{eqnarray}
Since $SS^{\rT}=I_d$, the matrix $S^{\rT}$ describes an isometry between $\mR^d$ and the subspace $S^{\rT}\mR^d\subset \mR^n$. Accordingly, the integrals in (\ref{H0t}) and (\ref{Ht}), as generalized functions \cite{V_2002}, are measures on this subspace with densities $\Psi$ and $\Ups$ (with respect to the $d$-dimensional Lebesgue measure on $S^{\rT}\mR^d$) and values in $\mC$ and $\mC^m$, respectively. Regardless of the particular form (\ref{h0tilde}) and (\ref{htilde}) of the perturbations $\wt{h}_0$ and $\wt{h}$ in (\ref{h0eps}) and (\ref{heps}), the perturbed GKSL generator (\ref{cL}) depends on $\eps$ in a quadratic fashion:
\begin{eqnarray}
\nonumber
   \cL_{\eps}(\xi)
   & := &
   i[h_0^{\eps},\xi]
     +
    \frac{1}{2}
    \big(
        [(h^{\eps})^{\rT},\xi]\Omega h^{\eps}  + (h^{\eps})^{\rT}\Omega [\xi,h^{\eps}]
    \big)\\
\label{cLeps}
    & = &
    \cL(\xi) + \eps \wt{\cL}(\xi)
    +
    \frac{\eps^2}{2}
    \big(
        [\wt{h}^{\rT},\xi]\Omega \wt{h}  + \wt{h}^{\rT}\Omega [\xi,\wt{h}]
    \big),
\end{eqnarray}
where $\cL:= \cL_0$ is the unperturbed generator.
Here, the parametric derivative $\wt{\cL}:= \d_{\eps}\cL_{\eps}\big|_{\eps = 0}$ is a linear superoperator  which acts on a system operator $\xi$ as
\begin{eqnarray}
\nonumber
   \wt{\cL}(\xi)
   &= &
   i[\wt{h}_0,\xi]\\
\nonumber
    & & +
     \frac{1}{2}
    \big(
        [\wt{h}^{\rT},\xi]\Omega MX
        +
        [X^{\rT},\xi]M^{\rT}\Omega \wt{h} \\
\label{cLtilde}
     &   & +
        \wt{h}^{\rT}\Omega M[\xi,X]
        +
        X^{\rT}M^{\rT}\Omega[\xi,\wt{h}]
    \big),
\end{eqnarray}
where $M$ %$M_1, \ldots, M_m$ are the rows of the coupling matrix $M$
is the coupling matrix of the nominal OQHO in (\ref{hhOQHO2}). In the case of the Weyl variations (\ref{h0tilde}) and (\ref{htilde}), the superoperator $\wt{\cL}$ depends linearly on the functions $\Psi$ and $\Ups$ which quantify the strength of such perturbations of the energy operators.

\subsection*{Example 3}

Suppose $\Ups=0$ (that is, the coupling operators in (\ref{heps}) remain unperturbed). Then the QSDE (\ref{dX}) takes the form \cite{V_2015c}
\begin{equation}
\label{dXhlin}
    \rd X
    =
    \Big(
        AX+2\eps i\Theta S^{\rT}\int_{\mR^d} \Psi(v) v\cW_{S^{\rT}v}\rd v
    \Big)\rd t + B\rd W,
\end{equation}
where the nonlinear dependence of the drift term on the system variables  comes from the nonquadratic part of the Hamiltonian in (\ref{h0eps}).
In the particular case (\ref{Xqp}) and (\ref{pot}), considered in Examples 1 and 2,  the Weyl variation of the potential energy leads to the Hamiltonian
\begin{equation}
\label{poteps}
    h_0
    =
    \frac{1}{2}(q^{\rT} \sK q + p^{\rT} \sM^{-1} p)
    +
    \eps
    \phi(q),
    \qquad
    \phi(q)
    :=
    \int_{\mR^d} \Psi(v) \re^{iv^{\rT}q}\rd v,
\end{equation}
which corresponds to (\ref{h0eps}) with the energy matrix $R$ given by  (\ref{Rqp}) and the matrix $S$ extracting the position variables from the vector $X$:
\begin{equation}
\label{SId}
    S
    =
    \left[
    \begin{array}{cc}
        I_d & 0
    \end{array}
    \right],
\end{equation}
with
\begin{equation}
\label{d}
    d
    :=
    \frac{n}{2}.
\end{equation}
This perturbation enters the QSDE (\ref{dXhlin}) through the momentum part in the form of an additional ``force'' term in the drift vector as
\begin{eqnarray}
\label{dq}
    \rd q
    & = &
    (A_{11} q + A_{12} p ) \rd t + B_1 \rd W,\\
\label{dp}
    \rd p
    & = &
    ( A_{21} q + A_{22} p  - \eps\phi'(q) ) \rd t + B_2 \rd W,
\end{eqnarray}
where
\begin{equation}
\label{phidash}
    \phi'(q)
    =
    i
    \int_{\mR^d}
    \Psi(v) v \re^{iv^{\rT} q}
    \rd v.
\end{equation}
Here, use is made of the CCR matrix $\Theta$ from (\ref{Thetaqp}), whereby $2\Theta S^{\rT} = {\small\left[\begin{array}{c} 0 \\ - I_d\end{array}\right]}$, and the matrices $A$ and $B$ in (\ref{A}) and (\ref{B}) are partitioned into blocks  $A_{jk} \in \mR^{d\x d}$  and $B_j \in \mR^{d\x m}$ as
\begin{equation}
\label{ABblocks}
    A
    :=
    \left[
    \begin{array}{ll}
      A_{11} & A_{12} \\
      A_{21} & A_{22}
    \end{array}
    \right],
    \qquad
    B
    :=
    \left[
    \begin{array}{l}
      B_1\\
      B_2
    \end{array}
    \right].
\end{equation}
Note that the form of the position QSDE (\ref{dq}) remains unperturbed.
The integrals in (\ref{poteps}) and (\ref{phidash}) lend themselves to closed-form evalution, for example, if $\Psi$ is a linear combination of quadratic-exponential functions which describe Gaussian-shaped  ``wells'' or ``bumps'' in the potential energy \cite{FM_2004} (see also \cite[Section 9]{V_2015c}). More precisely, suppose the function $\Psi$ is given by
\begin{equation}
\label{Psigauss}
    \Psi(v)
    :=
    \alpha \frac{\sqrt{\det \Lambda}}{(2\pi)^{d/2}}
    \re^{-iv^{\rT}\gamma-\frac{1}{2}\|v\|_{\Lambda}^2},
\end{equation}
where $\alpha$ is a real-valued constant, $\gamma \in \mR^d$, and $\Lambda$ is a real positive definite symmetric matrix of order (\ref{d}). The corresponding  integral in (\ref{poteps}) is the inverse Fourier transform of (\ref{Psigauss}):
\begin{equation}
\label{phiMorse}
    \phi(q)
    =
    \alpha
    \re^{
    -\frac{1}{2}
    \|q-\gamma\|_{\Lambda^{-1}}^2
    },
\end{equation}
whose right-hand side (up to a constant factor)   is  the Gaussian PDF with the mean vector $\gamma$ and covariance matrix $\Lambda$.
The parameter $\gamma$ of the potential $\phi$ in (\ref{phiMorse})
specifies the location of a centre of  attraction (if $\alpha<0$) or repulsion (if $\alpha>0$) in the position space $\mR^d$ with the stiffness matrix $\phi''(\gamma) = -\alpha\Lambda^{-1}$. The exponentially fast decay of this potential (together with its gradient) at infinity resembles the Morse potential \cite{M_1929}. The corresponding perturbation force term in (\ref{dp}) is obtained by a straightforward differentiation of (\ref{phiMorse}) with respect to $q$ as
$
    \phi'(q) = \phi(q) \Lambda^{-1}(\gamma-q)
$.
\hfill$\blacktriangle$

%%%%%%%%%%%%%%%%%%%%%%%%%%%%%%%%%%%%%%%%%%%%%%%%%%%%%%%%%%%%%%%%%%%%%%%%%%%%%%%%%%%%%%%%%%%%%%%%%%%%%%%%%%%%%%%
The nonquadratic perturbations of the Hamiltonian, considered in the above example,  are applicable to the modelling of  open quantum systems with multiextremum potential energy landscapes in the form
$$
    \phi(q)
    =
    \sum_{k=1}^N
    \alpha_k
    \exp
    \Big(
        -\frac{1}{2}
        \|q-\gamma_k\|_{\Lambda_k^{-1}}^2
    \Big),
$$
which extends (\ref{phiMorse}). In addition to being ubiquitous in macromolecular dynamics simulations, multiextremum potentials are also of interest from the viewpoint of adiabatic quantum computing \cite{AVKLLR_2005}, where they are represented by the terminal Hamiltonian of a quantum system, towards  which the initial Hamiltonian  is adiabatically evolved together with its ground state. In fact, the efficiency of the quantum annealing protocols in comparison with the classical optimization algorithms is particularly important
for the minimization of complicated functions of many variables.

%%%%%%%%%%%%%%%%%%%%%%%%%%%%%%%%%%%%%%%%%%%%%%%%%%%%%%%%%%%%%%%%%%%%%%%%%%%%%%%%%%%%%%%%%%%%%%%%%%%
\section{First-order correction of the quantum state}\label{sec:statevar}
%%%%%%%%%%%%%%%%%%%%%%%%%%%%%%%%%%%%%%%%%%%%%%%%%%%%%%%%%%%%%%%%%%%%%%%%%%%%%%%%%%%%%%%%%%%%%%%%%%%

In accordance with the perturbed energy operators in (\ref{h0eps}) and (\ref{heps}), the operators $\fA$, $\fF$ in (\ref{fA}), (\ref{mhodot}), the QCF $\Phi$  and the QPDF $\mho$ acquire dependence  on the parameter $\eps$ and will be denoted by $\fA_{\eps}$, $\fF_{\eps}$, $\Phi_{\eps}$, $\mho_{\eps}$, respectively (the subscript $(\cdot)_0$ for the unperturbed quantities at $\eps=0$ will be omitted for brevity). Furthermore, in view of (\ref{V}) and (\ref{K}), the integral operators $\fA_{\eps}$, $\fF_{\eps}$ depend on $\eps$ in a quadratic fashion, inheriting such dependence from the GKSL generator $\cL_{\eps}$ in  (\ref{cLeps}).
%Since the QPDF is related by the Fourier transform to the QCF, we will be mainly concerned  with the latter.
The QCF and QPDF  of the perturbed quantum system admit the asymptotic series expansions
$$
    \Phi_{\eps}
    =
    \sum_{r=0}^{+\infty}
    \frac{\eps^r}{r!}
    \Phi^{(r)},
    \qquad
    \mho_{\eps}
    =
    \sum_{r=0}^{+\infty}
    \frac{\eps^r}{r!}
    \mho^{(r)},
$$
where the functions
\begin{equation}
\label{Phir}
    \Phi^{(r)}
    :=
    \d_{\eps}^r
    \Phi_{\eps}
    \big|_{\eps = 0},
    \qquad
    \mho^{(r)}
    :=
    \d_{\eps}^r
    \mho_{\eps}
    \big|_{\eps = 0}
\end{equation}
satisfy the recurrent IDEs
\begin{eqnarray}
\label{Phidotr}
    \d_t \Phi^{(r)}
    & = &
    \sum_{k=0}^{\min(r,2)}
    \left(\begin{array}{c}
    r\\
    k
    \end{array}
    \right)
    \fA^{(k)}(\Phi^{(r-k)}),\\
\label{mhodotr}
    \d_t \mho^{(r)}
    & = &
    \sum_{k=0}^{\min(r,2)}
    \left(
    \begin{array}{c}
    r\\
    k
    \end{array}
    \right)
    \fF^{(k)}(\mho^{(r-k)})
\end{eqnarray}
for all $r=0,1,2,\ldots$. These are obtained by applying the Leibniz product rule to the differentiation of the equations (\ref{Phidot}) and (\ref{mhodot}) with respect to $\eps$, starting
from the PDEs (\ref{PhiPDE}) and (\ref{mhoPDE}) for the QCF $\Phi^{(0)} = \Phi$ and the QPDF $\mho^{(0)} = \mho$ of the unperturbed oscillator. Here, $    {\small\begin{pmatrix} r\\ k \end{pmatrix}}:= \frac{r!}{k!(r-k)!} $ is the binomial coefficient, and use is made of the parametric  derivatives
\begin{equation}
\label{fAr}
    \fA^{(r)}
    :=
    \d_{\eps}^r
    \fA_{\eps}
    \big|_{\eps = 0},
\qquad
    \fF^{(r)}
    :=
    \d_{\eps}^r
    \fF_{\eps}
    \big|_{\eps = 0},
\end{equation}
which vanish for all $r>2$ in view of the quadratic dependence of the operators $\fA_{\eps}$ and $\fF_{\eps}$ on $\eps$, thus leading to the band structure of the equations (\ref{Phidotr}) and (\ref{mhodotr}). This structure can be represented symbolically in vector-matrix form as
\begin{eqnarray*}
    \d_t
    \left[
    \begin{array}{c}
    \Phi\\
    \Phi^{(1)}\\
    \Phi^{(2)}\\
    \Phi^{(3)}\\
    \Phi^{(4)}\\
    \vdots
    \end{array}
    \right]
    & = &
    \left[
    \begin{array}{cccccc}
    \fA         & 0     & 0 & 0 &0 & \ldots\\
    \fA^{(1)}   & \fA   & 0 & 0 &0 & \ldots\\
    \fA^{(2)}   & 2\fA^{(1)} & \fA   & 0 & 0 &\ldots\\
    0 & 3\fA^{(2)}   & 3\fA^{(1)} & \fA   & 0 & \ldots\\
    0 & 0 & 6\fA^{(2)}   & 4\fA^{(1)} & \fA   & \ldots\\
    \vdots & \vdots& \vdots & \vdots & \vdots & \vdots
    \end{array}
    \right]
    \left[
    \begin{array}{c}
    \Phi\\
    \Phi^{(1)}\\
    \Phi^{(2)}\\
    \Phi^{(3)}\\
    \Phi^{(4)}\\
    \vdots
    \end{array}
    \right], \\
    \d_t
    \left[
    \begin{array}{c}
    \mho\\
    \mho^{(1)}\\
    \mho^{(2)}\\
    \mho^{(3)}\\
    \mho^{(4)}\\
    \vdots
    \end{array}
    \right]
    & = &
    \left[
    \begin{array}{cccccc}
    \fF         & 0     & 0 & 0 &0 & \ldots\\
    \fF^{(1)}   & \fF   & 0 & 0 &0 & \ldots\\
    \fF^{(2)}   & 2\fF^{(1)} & \fF   & 0 & 0 &\ldots\\
    0 & 3\fF^{(2)}   & 3\fF^{(1)} & \fF   & 0 & \ldots\\
    0 & 0 & 6\fF^{(2)}   & 4\fF^{(1)} & \fF   & \ldots\\
    \vdots & \vdots& \vdots & \vdots & \vdots & \vdots
    \end{array}
    \right]
    \left[
    \begin{array}{c}
    \mho\\
    \mho^{(1)}\\
    \mho^{(2)}\\
    \mho^{(3)}\\
    \mho^{(4)}\\
    \vdots
    \end{array}
    \right].
\end{eqnarray*}
In what follows, we will consider only the first-order  infinitesimal perturbations of the quantum state in (\ref{Phir}) described by
\begin{equation}
\label{Phi'mho'}
    \wt{\Phi}
    := \Phi^{(1)},
    \qquad
    \wt{\mho}
    :=
    \mho^{(1)}.
\end{equation}
Their evolution is governed by
\begin{eqnarray}
\label{Phi'dot}
    \d_t \wt{\Phi}
    & = &
    \fA(\wt{\Phi}) + \fB(\Phi),\\
\label{mho'dot}
    \d_t \wt{\mho}
    & = &
    \fF(\wt{\mho}) + \fG(\mho),
\end{eqnarray}
with zero initial conditions $
    \wt{\Phi}(0,\cdot) = 0$ and $
    \wt{\mho}(0,\cdot) = 0
$, where $\Phi$  and $\mho$ are the QCF  and QPDF for the unperturbed OQHO satisfying the PDEs (\ref{PhiPDE}) and (\ref{mhoPDE}). Here, the operators $\fA$ and $\fF$ from (\ref{fAr}) act on the functions (\ref{Phi'mho'}), and their parametric derivatives
\begin{equation}
\label{fA'fF'}
    \fB
    :=
    \fA^{(1)},
    \qquad
    \fG
     :=
    \fF^{(1)}
    =
    \mF \fB\mF^{-1}
\end{equation}
are related by the same unitary similarity transformation as in (\ref{mhodot}).
The operator $\fB$ is related to the superoperator $\wt{\cL}$ in (\ref{cLtilde})  through the QCF $\Phi$ in (\ref{Phi}) and the Weyl operator $\cW_u$ in (\ref{cW}) by
\begin{equation}
\label{fBPhi}
    \fB(\Phi(t,\cdot))(u) = \bE \wt{\cL}(\cW_u(t)),
    \qquad
    t\> 0,\
    u \in \mR^n.
\end{equation}
The following lemma establishes the structure of $\fB$ as a perturbation of the integral operator $\fA$ in (\ref{fA}) rather than through the averaging in (\ref{fBPhi}).

%%%%%%%%%%%%%%%%%%%%%%%%%%%%%%%%%%%%%%%%%%%%%%%%%%%%%%%%%%%%%%%%%%%%%%%%%%%%%%%%%%%%%%%%%%%%%%%%%%%%%%%%%%%%%%
\begin{lem}
\label{lem:fB}
Suppose the Hamiltonian and coupling operators of the OQHO are subject to the Weyl variations (\ref{h0eps})--(\ref{htilde}). Then $\fB$ in (\ref{fA'fF'}) is an integro-differential operator which acts on a continuously differentiable function $\varphi: \mR^n \to \mC$ as
\begin{eqnarray}
\nonumber
  \fB(\varphi)(u)
    &=&
    -2
    \int_{\mR^d}
    \sin(u^{\rT}\Theta S^{\rT}w)\Psi(w) \varphi(u+S^{\rT} w)\rd w \\
\nonumber
    & & +
    2i
    \int_{\mR^d}
    \Ups(w)^{\rT}
    {\small
    \left[\begin{array}{cc}
        \sin(u^{\rT}\Theta S^{\rT}w) M\Theta (u+S^{\rT}w) + K(u, S^{\rT}w)M\Theta u
        &
        \sin(u^{\rT}\Theta S^{\rT}w) JM
    \end{array}
    \right]}\\
\label{fBphi}
    & &\x
    {\small
    \left[
    \begin{array}{c}
        \varphi(u+S^{\rT}w)\\
        \varphi'(u+S^{\rT}w)
    \end{array}
    \right]}
    \rd w,
    \qquad
    u \in \mR^n.
\end{eqnarray}
\hfill$\square$
\end{lem}
%%%%%%%%%%%%%%%%%%%%%%%%%%%%%%%%%%%%%%%%%%%%%%%%%%%%%%%%%%%%%%%%%%%%%%%%%%%%%%%%%%%%%%%%%%%%%%%%%%%%%%%%%%%%%%
\begin{proof}
The kernel function of the operator $\fB$ can be
obtained by substituting (\ref{H0eps}) and (\ref{Heps}) into (\ref{V}) and differentiating the result
\begin{eqnarray}
\nonumber
  V_{\eps}(u,v)
  &:= &
    -2        \sin(u^{\rT}\Theta v)H_0^{\eps}(v)\\
\label{Veps}
   &  & -2 \int_{\mR^n}
            \sin(u^{\rT}\Theta s)H^{\eps}(s)^{\rT}
            K(u+s,s-v)
        H^{\eps}(v-s) \rd s
\end{eqnarray}
with respect to $\eps$:
\begin{eqnarray}
\nonumber
  \wt{V}(u,v)
   &:= &
  \d_{\eps}V_{\eps}(u,v)\big|_{\eps=0}\\
\nonumber
  &=&
    -2        \sin(u^{\rT}\Theta v)\wt{H}_0(v)\\
\nonumber
   &  & -2 i\int_{\mR^n}
            \sin(u^{\rT}\Theta s)\delta'(s)^{\rT}M^{\rT}
            K(u+s,s-v)
        \wt{H}(v-s) \rd s\\
\label{V'}
   &  & -2 i\int_{\mR^n}
            \sin(u^{\rT}\Theta s)\wt{H}(s)^{\rT}
            K(u+s,s-v)
        M\delta'(v-s) \rd s,
\end{eqnarray}
where the function $K$ is given by (\ref{K}) and $\wt{H}_0$ and $\wt{H}$ are the measures in (\ref{H0t}) and (\ref{Ht}). The integrals in (\ref{V'}) originate from differentiating the integral in (\ref{Veps}).
The integral operator $\fB$ with the kernel function $\wt{V}$ acts on a function $\varphi: \mR^n \to \mC$ as
\begin{eqnarray}
\nonumber
  \fB(\varphi)(u)
    &=& \int_{\mR^n}
  \wt{V}(u,v)\varphi(u+v)
  \rd v\\
\nonumber
  &=&
    -2
    \int_{\mR^n}
    \sin(u^{\rT}\Theta v)\wt{H}_0(v)\varphi(u+v)\rd v \\
\label{fB}
   & & +
    \fB_1(\varphi)(u) + \fB_2(\varphi)(u).
\end{eqnarray}
Here, the integral operators $\fB_1$ and $\fB_2$ are associated with the Weyl variations of the system-field coupling operators as
\begin{eqnarray}
\nonumber
    \fB_1(\varphi)(u)
    &:=
    &
     -2 i
     \int_{\mR^n}
            \sin(u^{\rT}\Theta s)\delta'(s)^{\rT}M^{\rT}
        \Big(
            \int_{\mR^n}
            \varphi(u+v)K(u+s,s-v)
        \wt{H}(v-s)
        \rd v
        \Big) \rd s\\
\nonumber
    &=
     & 2 i
     \div_s
      \Big(
            \sin(u^{\rT}\Theta s)
            M^{\rT}
            \int_{\mR^n}
            \varphi(u+v)
            K(u+s,s-v)
            \wt{H}(v-s)
            \rd v
      \Big)
      \Big|_{s=0}\\
\label{fB1}
    &=
     & 2i u^{\rT}\Theta
            M^{\rT}
            \int_{\mR^n}
            \varphi(u+v)
            K(u,-v)
            \wt{H}(v)
            \rd v,\\
\nonumber
    \fB_2(\varphi)(u)
     &:= &
     -2 i\int_{\mR^n}
            \sin(u^{\rT}\Theta s)
            \wt{H}(s)^{\rT}
            \Big(
                \int_{\mR^n}
                \varphi(u+v)
            K(u+s,s-v)
            M\delta'(v-s)
            \rd v
            \Big)
        \rd s\\
\nonumber
     &= &
     2 i\int_{\mR^n}
            \sin(u^{\rT}\Theta s)
            \wt{H}(s)^{\rT}
            \div_v
            \big(
                \varphi(u+v)
            K(u+s,s-v)
            M
            \big)
            \big|_{v=s}
        \rd s\\
\label{fB2}
     &= &
     2 i\int_{\mR^n}
            \sin(u^{\rT}\Theta s)
            \wt{H}(s)^{\rT}
            (
                JM\varphi'(u+s)
                +
                \varphi(u+s) M\Theta (u+s)
            )
        \rd s,
\end{eqnarray}
where the computations are similar to (\ref{fAphi}) and also employ the relation (\ref{fdeldash}) together with the structure of the function $K$ in (\ref{K}). From (\ref{fB1}) and (\ref{fB2}), it follows that
\begin{eqnarray}
\nonumber
    \hskip-3mm \fB_1(\varphi)(u)
    +
    \fB_2(\varphi)(u)
    &=&
    2 i
     u^{\rT}\Theta
            M^{\rT}
            \int_{\mR^n}
            \varphi(u+v)
            (-\sin(u^{\rT}\Theta v) I_m + \cos(u^{\rT}\Theta v) J)
            \wt{H}(v)
            \rd v\\
\nonumber
    & &+
     2 i\int_{\mR^n}
            \sin(u^{\rT}\Theta v)
            \wt{H}(v)^{\rT}
            \big(
            J
            M\varphi'(u+v)
            +
            \varphi(u+v)
            M\Theta (u+v)
            \big)
        \rd v\\
\nonumber
    &= &
    2i
    \int_{\mR^n}
    \wt{H}(v)^{\rT}
    \left[
    \begin{array}{cc}
        \sin(u^{\rT}\Theta v) M\Theta (2u+v) + \cos(u^{\rT}\Theta v)JM\Theta u
        &
        \sin(u^{\rT}\Theta v) JM
    \end{array}
    \right]
    \left[
    \begin{array}{c}
        \varphi(u+v)\\
        \varphi'(u+v)
    \end{array}
    \right]
    \rd v\\
\label{fB12}
    &= &
    2i
    \int_{\mR^n}
    \wt{H}(v)^{\rT}
    \left[
    \begin{array}{cc}
        \sin(u^{\rT}\Theta v) M\Theta (u+v) + K(u,v)M\Theta u
        &
        \sin(u^{\rT}\Theta v) JM
    \end{array}
    \right]
    \left[
    \begin{array}{c}
        \varphi(u+v)\\
        \varphi'(u+v)
    \end{array}
    \right]
    \rd v.
\end{eqnarray}
In combination with the structure of the functions $\wt{H}_0$ and $\wt{H}$ as the measures in (\ref{H0t}) and (\ref{Ht}),  substitution of (\ref{fB12}) into (\ref{fB}) leads to (\ref{fBphi})   .
\end{proof}
%%%%%%%%%%%%%%%%%%%%%%%%%%%%%%%%%%%%%%%%%%%%%%%%%%%%%%%%%%%%%%%%%%%%%%%%%%%%%%%%%%%%%%%%%%%%%%%%%%%%%%%%%%%%%%

Application of the Duhamel formula \cite{E_1998,S_2008} %for the solutions of nonhomogeneous linear PDEs,
to the solutions of the initial value problems for the linear evolutionary equations (\ref{Phi'dot}) and (\ref{mho'dot}) (with zero initial conditions) leads to
\begin{eqnarray}
\nonumber
    \wt{\Phi}(t,\cdot)
    & = &
    \int_0^t
    (\re^{(t-s)\fA}\fB \re^{s\fA})(\Phi(0,\cdot))
    \rd s\\
\label{Phi't}
    & = &
    \int_0^t
    \re^{(t-s)\fA}(\fB(\Phi(s,\cdot)))
    \rd s,\\
\nonumber
    \wt{\mho}(t,\cdot)
    & = &
    \int_0^t
    (\re^{(t-s)\fF}\fG \re^{s\fF})(\mho(0,\cdot))
    \rd s\\
\label{mho't}
    & = &
    \int_0^t
    \re^{(t-s)\fF}(\fG(\mho(s,\cdot)))
    \rd s
\end{eqnarray}
for any time $t\>0$.
The linear operators on the right-hand sides of (\ref{Phi't}) and (\ref{mho't}) also involve the parametric derivatives of the  exponential maps \cite{K_1976} from (\ref{Phisol}) and (\ref{trans}).
The infinite-horizon limits of these solutions provide the first-order corrections
\begin{eqnarray}
\label{Phi'inf}
    \wt{\Phi}_*
    & := &
    \lim_{t\to +\infty}
    \wt{\Phi}(t,\cdot)
    =
    \int_0^{+\infty}
    \re^{t\fA}(\fB(\Phi_*))
    \rd t
    =
    -\fA^{-1}(\fB(\Phi_*)),\\
\label{mho'inf}
    \wt{\mho}_*
    & := &
    \lim_{t\to +\infty}
    \wt{\mho}(t,\cdot)
    =
    \int_0^{+\infty}
    \re^{t\fF}(\fG(\mho_*))
    \rd t
    =
    -\fF^{-1}(\fG(\mho_*)) = \mF(\wt{\Phi}_*)
\end{eqnarray}
to the Gaussian QCF and QPDF $\Phi_*$  and $\mho_*$ of the nominal invariant state described in Section~\ref{sec:gauss}, with the limits being understood similarly to the weak convergence \cite{B_1968}, so that the convergence in (\ref{Phi'inf}) is pointwise.  Here, $-\fA^{-1}\fB$ and $-\fF^{-1}\fG$ denote the operators which relate the steady-state  solutions of (\ref{Phi'dot}) and (\ref{mho'dot}) to $\Phi_*$ and $\mho_*$ subject to the  normalization constraints
\begin{equation}
\label{Phimhonorm}
  \wt{\Phi}_*(0)=0,
  \qquad
  \int_{\mR^n}
  \wt{\mho}_*(x)\rd x
  =0.
\end{equation}
Since the QPDF perturbation $\wt{\mho}_*$ in (\ref{mho'inf}) is related by the spatial Fourier transform $\mF$ to the QCF perturbation $\wt{\Phi}_*$, we will be concerned  mainly  with the latter. This will also allow advantage to be made of the fact that the operator exponential $\re^{t\fA}$, associated with the fundamental solution of the first-order linear PDE (\ref{PhiPDE}), is simpler than $\re^{t\fF}$ for the second-order PDE (\ref{mhoPDE}).

%%%%%%%%%%%%%%%%%%%%%%%%%%%%%%%%%%%%%%%%%%%%%%%%%%%%%%%%%%%%%%%%%%%%%%%%%%%%%%%%%%%%%%%%%%%%%%%%%%%
\begin{thm}
\label{th:QCFvar}
Suppose the unperturbed OQHO has a Hurwitz matrix $A$ in (\ref{A}), and its Hamiltonian and coupling operators are subject to the Weyl variations (\ref{h0eps})--(\ref{htilde}). Then the corresponding perturbation (\ref{Phi'inf})  of the invariant QCF is related by
\begin{equation}
\label{QCFvar}
    \wt{\Phi}_*(u)
    =
    \int_{\mR^d}
    \big(
        F(u,w) \Psi(w) + G(u,w)^* \Ups(w)
    \big)
    \rd w,
    \qquad
    u \in \mR^n,
\end{equation}
to the strength functions $\Psi$ and $\Ups$ of the Weyl variations.
Here, the functions $F: \mR^n\x \mR^d \to \mR$ and $G:= (G_k)_{1\< k \< m}: \mR^n\x \mR^d \to \mC^m$ are computed as
\begin{eqnarray}
\label{F0}
  F(u,w)
  &:= &
  2
      \re^{
        -\frac{1}{2} (\|u\|_P^2+\|S^{\rT}w\|_P^2)
    }
  \int_0^{+\infty}
  \Im
  \re^{-u^{\rT}\re^{tA}(P+i\Theta) S^{\rT}w}
  \rd t,\\
\nonumber
  G(u,w)
  &:= &
    -2i
          \re^{
        -\frac{1}{2} (\|u\|_P^2+\|S^{\rT}w\|_P^2)
    }
    \int_0^{+\infty}
    \big(
        \sin(u^{\rT}\re^{tA}\Theta S^{\rT}w) D (\re^{tA^{\rT}}u+S^{\rT}w)
        +
        K(\re^{tA^{\rT}}u,S^{\rT}w)M\Theta \re^{tA^{\rT}}u
    \big)\\
\label{G0}
   & & \x
        \re^{
        -
        u^{\rT}\re^{tA} P S^{\rT}w
    }
    \rd t,
\end{eqnarray}
where
\begin{equation}
\label{D}
    D := M\Theta-JMP.
\end{equation}
Here, the function $K$ is given by (\ref{K}), $M$ is the coupling matrix for the unperturbed oscillator  in (\ref{hhOQHO2}), and $P$ is the infinite-horizon controllability Gramian from (\ref{Siginv}).
\hfill$\square$
\end{thm}
%%%%%%%%%%%%%%%%%%%%%%%%%%%%%%%%%%%%%%%%%%%%%%%%%%%%%%%%%%%%%%%%%%%%%%%%%%%%%%%%%%%%%%%%%%%%%%%%%%%
\begin{proof}
By combining the second equality in (\ref{Phi'inf}) with the structure of the operator exponential $\re^{t\fA}$ in (\ref{Phisol}), it follows that
\begin{equation}
\label{QCFvar1}
    \wt{\Phi}_*(u)
    =
    \int_0^{+\infty}
    \fB(\Phi_*)(\re^{tA^{\rT}}u)
    \re^{-\frac{1}{2}\|u\|_{\Sigma(t)}^2}
    \rd t,
    \qquad
    u \in \mR^n,
\end{equation}
where use is made of the finite-horizon controllability Gramian  $\Sigma(t)$ from (\ref{Sigma}). Evaluation of the operator $\fB$ from (\ref{fBphi}) at the unperturbed invariant Gaussian QCF $\Phi_*$ in (\ref{Phi*})
leads to
\begin{eqnarray}
\nonumber
  \fB(\Phi_*)(u)
    &=&
    -2
    \int_{\mR^d}
    \sin(u^{\rT}\Theta S^{\rT}w)\Psi(w) \Phi_*(u+S^{\rT} w)\rd w \\
\nonumber
    && +
    2i
    \int_{\mR^d}
    \Ups(w)^{\rT}
    {\small
    \left[
    \begin{array}{cc}
        \sin(u^{\rT}\Theta S^{\rT}w) M\Theta (u+S^{\rT}w) + K(u,S^{\rT}w)M\Theta u
        &
        \sin(u^{\rT}\Theta S^{\rT}w) JM
    \end{array}
    \right]}
    {\small
    \left[
    \begin{array}{c}
        \Phi_*(u+S^{\rT}w)\\
        \Phi_*'(u+S^{\rT}w)
    \end{array}
    \right]}
    \rd w\\
\nonumber
    &=&
    -2
    \int_{\mR^d}
    \sin(u^{\rT}\Theta S^{\rT}w)\Psi(w) \Phi_*(u+S^{\rT} w)\rd w \\
\label{fBPhi*}
    && +
    2i
    \int_{\mR^d}
    \Ups(w)^{\rT}
    \big(
        \sin(u^{\rT}\Theta S^{\rT}w) D (u+S^{\rT}w) + K(u,S^{\rT}w)M\Theta u
    \big)
    \Phi_*(u+S^{\rT}w)
    \rd w,
\end{eqnarray}
where use is also made of the relation $\Phi_*'(v) = -\Phi_*(v)Pv$ for all $v\in \mR^n$ (following from the quadratic-exponential structure of $\Phi_*$) and the matrix $D$ from (\ref{D}). Substitution of (\ref{fBPhi*}) into (\ref{QCFvar1}) establishes the integral representation (\ref{QCFvar}) with the kernel functions
\begin{eqnarray}
\label{F}
  F(u,w)
  &:= &
  -2\int_0^{+\infty}
  \sin(u^{\rT}\re^{tA} \Theta S^{\rT} w)
  \Phi_*(\re^{tA^{\rT}} u + S^{\rT}w)
  \re^{-\frac{1}{2}\|u\|_{\Sigma(t)}^2}
  \rd t,\\
\nonumber
  G(u,w)
  &:= &
    -2i
    \int_0^{+\infty}
    \big(
        \sin(u^{\rT}\re^{tA}\Theta S^{\rT}w) D (\re^{tA^{\rT}}u+S^{\rT}w)
        +
        K(\re^{tA^{\rT}}u,S^{\rT}w)M\Theta \re^{tA^{\rT}}u
    \big)\\
\label{G}
   & & \x
        \Phi_*(\re^{tA^{\rT}}u+S^{\rT}w)
    \re^{-\frac{1}{2}\|u\|_{\Sigma(t)}^2}
    \rd t.
\end{eqnarray}
By using (\ref{Phi*}) and (\ref{PSig}), it follows that
\begin{eqnarray}
\nonumber
    \Phi_*(\re^{tA^{\rT}}u+v)\re^{-\frac{1}{2}\|u\|_{\Sigma(t)}^2}
    & = &
    \exp
    \Big(
        -\frac{1}{2} \|u\|_{\re^{tA} P \re^{tA^{\rT}} + \Sigma(t)}^2
        -
        u\re^{tA} P v
        -
        \frac{1}{2}\|v\|_P^2
    \Big)\\
\label{Phiexp}
    & = &
    \re^{
        -\frac{1}{2} (\|u\|_P^2+\|v\|_P^2)
        -
        u^{\rT}\re^{tA} P v
    }
\end{eqnarray}
for any $t\> 0$ and $u,v\in \mR^n$. Substitution of (\ref{Phiexp}) with $v:= S^{\rT}w$ into (\ref{F}) and (\ref{G}), leads to (\ref{F0}) and (\ref{G0}), where the convergence of the integrals is secured by the exponentially fast decay of the integrands since  the matrix $A$ is Hurwitz.
\end{proof}
%%%%%%%%%%%%%%%%%%%%%%%%%%%%%%%%%%%%%%%%%%%%%%%%%%%%%%%%%%%%%%%%%%%%%%%%%%%%%%%%%%%%%%%%%%%%%%%%%%%%

The functions $F$ and $G$ in (\ref{F0}) and (\ref{G0}) depend on the parameters of the unperturbed oscillator and can be identified with the Frechet derivatives $\d_{\Psi} \wt{\Phi}_*$ and $\d_{\Ups} \wt{\Phi}_*$ of the QCF perturbation $\wt{\Phi}_*$ with respect to the strength functions $\Psi$ and $\Ups$ of the Weyl variations in the energy operators. These derivatives describe the response of the invariant state of the quantum system in the framework of the infinitesimal perturbation analysis and allow the QCF corrections to be calculated for particular strength functions.

%%%%%%%%%%%%%%%%%%%%%%%%%%%%%%%%%%%%%%%%%%%%%%%%%%%%%%%%%%%%%%%%%%%%%%%%%%%%%%%%%%%%%%%%%%%%%%%%%%%
\subsection*{Example 4}
In the setting of Example~3 (with unperturbed system-field coupling operators), we will now compute the QCF correction in response to the Gaussian-shaped variation (\ref{phiMorse}) of the potential energy (see also Examples 1 and 2). By substituting its Fourier transform $\Psi$ from (\ref{Psigauss}) together with the influence function $F$ from (\ref{F0}) into (\ref{QCFvar}) and letting $\Ups = 0$, it follows that
\begin{eqnarray}
\nonumber
    \wt{\Phi}_*(u)
    & = &
    \int_{\mR^d}
    F(u,w) \Psi(w)
    \rd w\\
\nonumber
    & = &
    \frac{2\alpha \sqrt{\det \Lambda}}{(2\pi)^{d/2}}
      \re^{
        -\frac{1}{2} \|u\|_P^2
    }
    \int_{\mR^d}
    \Big(
      \re^{
        -\frac{1}{2} \|S^{\rT}w\|_P^2
    }
  \int_0^{+\infty}
  \Im
  \re^{-u^{\rT}\re^{tA}(P+i\Theta) S^{\rT}w}
  \rd t
  \Big)
    \re^{-iw^{\rT}\gamma-\frac{1}{2}\|w\|_{\Lambda}^2}
    \rd w\\
\nonumber
    & = &
    \frac{\alpha \sqrt{\det \Lambda}}{(2\pi)^{d/2}i}
      \re^{
        -\frac{1}{2} \|u\|_P^2
    }
  \int_0^{+\infty}
  \Big(
    \int_{\mR^d}
      \re^{
        -\frac{1}{2} \|w\|_{\Xi}^2
    }
  \Big(
    \re^{-u^{\rT}\re^{tA}(P+i\Theta) S^{\rT}w - iw^{\rT}\gamma}
    -
    \re^{-u^{\rT}\re^{tA}(P-i\Theta) S^{\rT}w-iw^{\rT}\gamma}
  \Big)
    \rd w
    \Big)
    \rd t\\
\nonumber
    & = &
    \frac{\alpha \sqrt{\det \Lambda}}{(2\pi)^{d/2}i}
      \re^{
        -\frac{1}{2} \|u\|_P^2
    }
  \int_0^{+\infty}
  \Big(
    \int_{\mR^d}
      \re^{
        -\frac{1}{2} \|w\|_{\Xi}^2
    }
  \Big(
    \re^{-w^{\rT}\sigma_-(t,u)}
    -
    \re^{-w^{\rT}\sigma_+(t,u)}
  \Big)
    \rd w
    \Big)
    \rd t\\
\label{QCFcorr}
    & = &
    \alpha i
    \sqrt{\frac{\det \Lambda}{\det \Xi}}
      \re^{
        -\frac{1}{2} \|u\|_P^2
    }
  \int_0^{+\infty}
  \Big(
  \re^{\frac{1}{2}\sigma_+(t,u)^{\rT}\Xi^{-1}\sigma_+(t,u)}
  -
  \re^{\frac{1}{2}\sigma_-(t,u)^{\rT}\Xi^{-1}\sigma_-(t,u)}
  \Big)
    \rd t.
\end{eqnarray}
Here, use is made of the moment-generating function
$$
    \frac{\sqrt{\det \Xi}}{(2\pi)^{d/2}}
    \int_{\mR^d}
    \re^{-\frac{1}{2}\|w\|_{\Xi}^2 + \sigma^{\rT} w}
    \rd w
    =
    \re^{\frac{1}{2}\sigma^{\rT}\Xi^{-1}\sigma},
    \qquad
    \sigma \in \mC^d,
$$
for a zero-mean Gaussian distribution  in $\mR^d$ with the covariance matrix $\Xi^{-1}$, where
\begin{equation}
\label{Xi}
  \Xi
  :=
  SPS^{\rT} + \Lambda
  =
  P_{11}+\Lambda
\end{equation}
is a real positive definite symmetric matrix of order $d$. In view of the structure of the matrix $S$ in (\ref{SId}), the block
$P_{11} \in \mR^{d\x d}$ of the matrix
\begin{equation}
\label{Pblocks}
    P
    :=
    \left[
    \begin{array}{cc}
      P_{11} & P_{12} \\
      P_{21} & P_{22}
    \end{array}
    \right]
\end{equation}
from (\ref{Siginv})
is associated with the position variables in accordance with the partitioning (\ref{Xqp}) and (\ref{ABblocks}).
Also, $\sigma_+,\sigma_-: \mR_+ \x \mR^n \to \mC^d$ are auxiliary functions given by
\begin{eqnarray}
\nonumber
  \sigma_{\pm}(t,u)
  & := &
  S(P\pm i\Theta) \re^{tA^{\rT}}u + i\gamma\\
\label{sig+-}
  & = &
  \left[
  \begin{array}{cc}
    P_{11} & P_{12} \pm \frac{i}{2}I_d
  \end{array}
  \right]
  \re^{tA^{\rT}}u + i\gamma,
\end{eqnarray}
where (\ref{Pblocks}) is combined with the structure of the CCR matrix $\Theta$ from  (\ref{Thetaqp}).
Since the functions in (\ref{sig+-}) converge to
a common limit $\lim_{t\to +\infty}\sigma_{\pm}(t,u) = \sigma_{\pm}(0,0) = i\gamma$ exponentially fast due to $A$ being Hurwitz, the integral in (\ref{QCFcorr}) is convergent for any $u\in \mR^n$.
\hfill$\blacktriangle$
%%%%%%%%%%%%%%%%%%%%%%%%%%%%%%%%%%%%%%%%%%%%%%%%%%%%%%%%%%%%%%%%%%%%%%%%%%%%%%%%%%%%%%%%%%%%%%%%%%%

The presence of double exponentials (such as $      \re^{
        -
        u^{\rT}\re^{tA} P S^{\rT}w
    }
$) in the integrals (\ref{F0}), (\ref{G0}) and (\ref{QCFcorr}) complicates their closed-form evaluation. Upper bounds for the influence functions
$F$ and $G$, which relate the QCF perturbation to the Weyl variations in Theorem~\ref{th:QCFvar}, will be discussed  in Section~\ref{sec:gain}. These functions
also encode more tractable first-order corrections  to the mean vector and the second-moment matrix for the system variables in the invariant (in general, non-Gaussian) state under such variations:
\begin{equation}
\label{muP}
    \wt{\mu}
    :=
    \d_{\eps}\bE_{\eps} X\big|_{\eps=0},
    \qquad
    \wt{P}
    :=
    \d_{\eps}\bE_{\eps} (XX^{\rT})\big|_{\eps=0},
\end{equation}
where $\bE_{\eps}(\cdot)$ denotes the quantum expectation over the invariant state of the perturbed system. Note that the constant imaginary part $\Im \bE_{\eps} (XX^{\rT}) = \Theta$ does not contribute to the parametric derivative in (\ref{muP}).    Therefore,
the corresponding truncations of the  asymptotic series expansions for these moments in the perturbed invariant state take the form
\begin{equation}
\label{muPexp}
    \bE_{\eps} X
    =
    \eps\wt{\mu} + \ldots,
    \qquad
    \Re \bE_{\eps} (XX^{\rT})
    =
    P + \eps\wt{P} + \ldots,
\end{equation}
where use is made of the zero mean vector and the quantum covariance matrix $P+i\Theta$ of the unperturbed invariant Gaussian state. In view of the first of the relations (\ref{muPexp}), the system variables in the perturbed invariant state acquire nonzero mean values.

%%%%%%%%%%%%%%%%%%%%%%%%%%%%%%%%%%%%%%%%%%%%%%%%%%%%%%%%%%%%%%%%%%%%%%%%%%%%%%%%%%%%%%%%%%%%%%%%%%%
\begin{thm}
\label{th:muP}
Suppose the unperturbed OQHO has a Hurwitz matrix $A$ in (\ref{A}), and its Hamiltonian and coupling operators are subject to the Weyl variations (\ref{h0eps})--(\ref{htilde}). Then the first-order corrections (\ref{muP}) for the mean vector and the second-moment matrix of the perturbed invariant state are related to the strength functions $\Psi$ and $\Ups$ of the Weyl variations by
\begin{eqnarray}
\label{mutilde0}
    \wt{\mu}
    & = &
    -2i
    A^{-1}
    \Theta
    \int_{\mR^d}
    \re^{
        -\frac{1}{2} \|S^{\rT}w\|_P^2
    }
    \big(
        \Psi(w)
        S^{\rT}w
        -
        i
        \big(
            S^{\rT}ww^{\rT}S
            D^{\rT}
            +
            M^{\rT}J
        \big)
        \Ups(w)
    \big)
    \rd w,\\
\label{Ptilde0}
    \wt{P}
    & = &
    -4
    \int_{\mR^d}
    \re^{
        -\frac{1}{2} \|S^{\rT}w\|_P^2
    }
    \Big(
        \Psi(w)Q(w)
        +
        i
        \sum_{k=1}^m
        \Ups_k(w)
        R_k(w)
    \Big)
    \rd w.
\end{eqnarray}
Here, the matrix-valued  functions $Q:\mR^d \to \mR^{n\x n}$ and $R_1, \ldots, R_m: \mR^d \to \mR^{n\x n}$ are, respectively, quadratic and cubic polynomials defined as the unique solutions of the ALEs
\begin{eqnarray}
\label{QALE}
    A Q(w) &+ & Q(w)A^{\rT} +
    \bS(\Theta S^{\rT}ww^{\rT}S P)
     =0,\\
\label{RALE}
    A R_k(w) & + & R_k(w)A^{\rT}
     +
    \bS
    \big(
        \Theta S^{\rT}w D_{k\bullet} (I_n - S^{\rT}ww^{\rT}SP)
        +
        \Theta S^{\rT}w (M\Theta)_{k\bullet}
        -
        PS^{\rT}w (JM\Theta)_{k\bullet}
    \big)
     = 0
\end{eqnarray}
for all $w\in \mR^d$, where $\bS(N):= \frac{1}{2}(N+N^{\rT})$  denotes the symmetrizer of square matrices, the matrix $D$ is given by (\ref{D}), and $(\cdot)_{k\bullet} $ is the $k$th row of a matrix.
\hfill$\square$
\end{thm}
%%%%%%%%%%%%%%%%%%%%%%%%%%%%%%%%%%%%%%%%%%%%%%%%%%%%%%%%%%%%%%%%%%%%%%%%%%%%%%%%%%%%%%%%%%%%%%%%%%%
\begin{proof}
By differentiating the representation (\ref{QCFvar}) for the QCF perturbation $\wt{\Phi}_*$ in accordance with (\ref{EX}), (\ref{EXX}) and (\ref{muP}), it follows that
\begin{eqnarray}
\label{mutilde}
    \wt{\mu}
    & = &
    -i\wt{\Phi}_*'(0)
    =
    -i
    \int_{\mR^d}
    \big(
        \Psi(w)\d_uF(u,w)\big|_{u=0}  + \big(\d_uG(u,w)\big|_{u=0}\big)^* \Ups(w)
    \big)
    \rd w,\\
\label{Ptilde}
    \wt{P}
    & = &
    -\wt{\Phi}_*''(0)
    =
    -
    \int_{\mR^d}
    \Big(
        \Psi(w)\d_u^2F(u,w)\big|_{u=0}  + \sum_{k=1}^m\Ups_k(w)\overline{\d_u^2G_k(u,w)}\big|_{u=0}
    \Big)
    \rd w,
\end{eqnarray}
where $\d_uG(u,w)$ is the Jacobian $(m\x n)$-matrix of $G$, and  $\d_u^2G_k(u,w)$ is the Hessian matrix of the $k$th coordinate function of $G$ with respect to $u\in \mR^n$. Now, (\ref{F0}) leads to
\begin{eqnarray}
\nonumber
    \d_u F(u,w)\big|_{u=0}
    & = &
    \Big(
    -F(u,w)Pu
    -
    2
    \re^{
        -\frac{1}{2} (\|u\|_P^2+\|S^{\rT}w\|_P^2)
    }
  \int_0^{+\infty}
  \Im
  \big(
  \re^{-u^{\rT}\re^{tA}(P+i\Theta) S^{\rT}w}
  \re^{tA}(P+i\Theta) S^{\rT}w
  \big)
  \rd t
  \Big)
  \Big|_{u=0}\\
\nonumber
  & = &
  -2
    \re^{
        -\frac{1}{2} \|S^{\rT}w\|_P^2
    }
    \int_0^{+\infty}
    \re^{tA} \rd t\,
     \Theta S^{\rT}w  \\
\label{F'}
  & = &
  2
    \re^{
        -\frac{1}{2} \|S^{\rT}w\|_P^2
    }
    A^{-1}
    \Theta S^{\rT}w,\\
\nonumber
    \d_u^2 F(u,w)\big|_{u=0}
    & = &
    2
    \re^{
        -\frac{1}{2} \|S^{\rT}w\|_P^2
    }
  \int_0^{+\infty}
  \re^{tA}
  \Im
  \big(
  (P+i\Theta) S^{\rT}ww^{\rT}S (P-i\Theta)
  \big)
  \re^{tA^{\rT}}
  \rd t\\
\nonumber
  & = &
  2
    \re^{
        -\frac{1}{2} \|S^{\rT}w\|_P^2
    }
      \int_0^{+\infty}
  \re^{tA}
  \big(
  \Theta S^{\rT}ww^{\rT}S P
  -
  P S^{\rT}ww^{\rT}S \Theta
  \big)
  \re^{tA^{\rT}}
  \rd t\\
\label{F''}
  & = &
  4
    \re^{
        -\frac{1}{2} \|S^{\rT}w\|_P^2
    }
  Q(w),
\end{eqnarray}
where (similarly to (\ref{Siginv}) and (\ref{PALE})) the last integral is the solution $Q(w)$ of the ALE (\ref{QALE}) which depends parametrically on $w\in \mR^d$ in a quadratic fashion. Also, (\ref{G0}) yields
\begin{eqnarray}
\nonumber
    \d_u G(u,w)\big|_{u=0}
   & = &
    2i
          \re^{
        -\frac{1}{2} \|S^{\rT}w\|_P^2
    }
    \int_0^{+\infty}
    \big(
    D S^{\rT}ww^{\rT}S
    -
    JM
    \big)
    \Theta \re^{tA^{\rT}}
    \rd t\\
\label{G'}
    & = &
    -2i
          \re^{
        -\frac{1}{2} \|S^{\rT}w\|_P^2
    }
    \big(
    D S^{\rT}ww^{\rT}S
    -
    JM
    \big)
    \Theta
    A^{-\rT},\\
\nonumber
  \d_u^2 G_k(u,w)\big|_{u=0}
  &= &
    -2i
          \re^{
        -\frac{1}{2} \|S^{\rT}w\|_P^2
    }\\
\nonumber
   & & \x
    \int_0^{+\infty}
    \d_u^2
    \Big(
    \big(
        \sin(u^{\rT}\re^{tA}\Theta S^{\rT}w)
        D_{k\bullet} (\re^{tA^{\rT}}u+S^{\rT}w)\\
\nonumber
       & & +
        (K(\re^{tA^{\rT}}u,S^{\rT}w)
        M\Theta)_{k\bullet} \re^{tA^{\rT}}u
    \big)
        \re^{
        -
        u^{\rT}\re^{tA} P S^{\rT}w
    }
    \Big)\Big|_{u=0}
    \rd t\\
\nonumber
    &= &
    -4i
          \re^{
        -\frac{1}{2} \|S^{\rT}w\|_P^2
    }\\
\nonumber
    &&\x
    \int_0^{+\infty}
    \re^{tA}
    \bS
    \big(
        \Theta S^{\rT}w D_{k\bullet} (I_n - S^{\rT}ww^{\rT}SP)
        +
        \Theta S^{\rT}w (M\Theta)_{k\bullet}
        -
        PS^{\rT}w (JM\Theta)_{k\bullet}
    \big)
    \re^{tA^{\rT}}
    \rd t\\
\label{G''}
    &= &
    -4i
          \re^{
        -\frac{1}{2} \|S^{\rT}w\|_P^2
    }R_k(w)
\end{eqnarray}
for all $k=1, \ldots, m$,
where $A^{-\rT}:= (A^{-1})^{\rT}$, and the matrix $R_k(w)$ is the solution of the ALE (\ref{RALE}) which is a cubic function of $w\in \mR^d$.
Substitution of (\ref{F'}) and (\ref{G'}) into (\ref{mutilde}) leads to (\ref{mutilde0}). Similarly, substitution of (\ref{F''}) and (\ref{G''}) into (\ref{Ptilde}) establishes (\ref{Ptilde0}).
\end{proof}
%%%%%%%%%%%%%%%%%%%%%%%%%%%%%%%%%%%%%%%%%%%%%%%%%%%%%%%%%%%%%%%%%%%%%%%%%%%%%%%%%%%%%%%%%%%%%%%%%%%

Note that in the setting of Example~3 (where only the Hamiltonian is perturbed), the mean vector correction $\wt{\mu}$ can be obtained by averaging the drift of the QSDE (\ref{dXhlin}) over the unperturbed invariant Gaussian state with the QCF $\Phi_*$ in (\ref{Phi*}) as
\begin{eqnarray}
\nonumber
    0
    &= &
    A\wt{\mu}+2 i\Theta S^{\rT}\int_{\mR^d} \Psi(v) v\bE \cW_{S^{\rT}v}\rd v\\
\nonumber
    &=&
    A\wt{\mu}+2 i\Theta S^{\rT}\int_{\mR^d} \Psi(v) \Phi_*(S^{\rT}v)v\rd v\\
\nonumber
    &=&
    A\wt{\mu}+2 i\Theta S^{\rT}\int_{\mR^d} \Psi(v) \re^{-\frac{1}{2}\|S^{\rT}v\|_P^2}v\rd v,
\end{eqnarray}
which leads to
$$
    \wt{\mu}= -2 iA^{-1}\Theta S^{\rT}\int_{\mR^d} \Psi(v) \re^{-\frac{1}{2}\|S^{\rT}v\|_P^2}v\rd v
$$
in accordance with (\ref{mutilde0}) with $\Ups = 0$. This is yet  another demonstration of the origin of the quadratic-exponential factors in the correction terms.

%%%%%%%%%%%%%%%%%%%%%%%%%%%%%%%%%%%%%%%%%%%%%%%%%%%%%%%%%%%%%%%%%%%%%%%%%%%%%%%%%%%%%%%%%%%%%%%%%%%%
\section{Norm bounds for the perturbation operators}\label{sec:gain}
%%%%%%%%%%%%%%%%%%%%%%%%%%%%%%%%%%%%%%%%%%%%%%%%%%%%%%%%%%%%%%%%%%%%%%%%%%%%%%%%%%%%%%%%%%%%%%%%%%%%

For what follows, we will now specify a class of the strength functions $\Psi$ and $\Ups$ of the Weyl variations (\ref{h0tilde}) and (\ref{htilde}) in the energy operators. More precisely, they are assumed to be square integrable with a weight in the sense that
\begin{equation}
\label{PsiUpsgood}
    \sn \Psi \sn_{\theta}<+\infty,
    \qquad
    \sn \Ups\sn_{\theta}<+\infty
\end{equation}
for some $\theta>0$. The parameter $\theta$ specifies the following weighted $L^2$-norm \cite{V_2015c} for a scalar or vector-valued function $\varphi$ on $\mR^d$:
\begin{equation}
\label{sn}
    \sn \varphi \sn_{\theta}
     :=
    \sqrt{
    \int_{\mR^d}
    |\varphi(w)|^2
    \re^{\theta |w|^2}
    \rd w
    }.
\end{equation}
The fulfillment of (\ref{PsiUpsgood}) ensures that $\Psi$ and $\Ups$ are not only square integrable, but are also absolutely integrable in view of the relations
\begin{eqnarray}
\nonumber
    \|\varphi\|_1
    & := &
    \int_{\mR^d}
    |\varphi(w)|
    \rd w\\
\nonumber
    & = &
    \int_{\mR^d}
    \re^{-\frac{1}{2}\theta |w|^2}
    |\varphi(w)|
    \re^{\frac{1}{2}\theta |w|^2}
    \rd w \\
\label{sn1}
    & \< &
    \sn 1\sn_{-\theta}
    \sn \varphi\sn_{\theta},
\end{eqnarray}
with
$$
    \sn 1\sn_{-\theta}
    =
    \sqrt{
    \int_{\mR^d}
    \re^{-\theta |w|^2}
    \rd w}
    =
    \Big(
        \frac{\pi}{\theta}
    \Big)^{d/4},
$$
which follow from (\ref{sn}) and the Cauchy-Bunyakovsky-Schwarz inequality. Note that the limiting case  $\sn\cdot\sn_0$ is the standard $L^2$-norm $\|\cdot\|_2$.

Now, the relations (\ref{mutilde}) and (\ref{Ptilde}) in the proof of Theorem~\ref{th:muP}  show that the functions (\ref{F'})--(\ref{G''}) specify the Frechet derivatives of the first-order corrections $\wt{\mu}$ and $\wt{P}$ for the mean vector and the covariance matrix with respect to $\Psi$ and $\Ups$.
In particular, the function $w\mapsto i\d_u F(u,w)\big|_{u=0}$, computed in (\ref{F'}), describes  the Frechet derivative  $\d_{\Psi} \wt{\mu}$ with respect to square integrable strength functions $\Psi$ of the Weyl variations (\ref{h0tilde}) in the system Hamiltonian. Similarly to (\ref{sn1}), the sensitivity of the mean vector of the perturbed invariant quantum state to $\Psi$ with the norm $\sn \Psi\sn_{\theta}$ can be quantified by
\begin{eqnarray}
\nonumber
    \sup_{\sn \Psi\sn_{\theta}\< 1,\ \Ups=0} |\wt{\mu}|
    & \< &
    \sn\d_{\Psi} \wt{\mu}\sn_{-\theta}\\
\nonumber
    & = &
    \sqrt{
    \int_{\mR^d}
    |\d_u F(u,w)|_{u=0}|^2
    \re^{-\theta |w|^2}
    \rd w}\\
\nonumber
    & = &
    2
    \sqrt{
    \int_{\mR^d}
        \re^{
        -\|S^{\rT}w\|_P^2
    }
    |
    A^{-1}
    \Theta S^{\rT}w|^2
    \re^{-\theta |w|^2}
    \rd w}\\
\nonumber
    & = &
    \frac{2\pi^{d/4}}{\sqrt[4]{\det(\theta I_d + SPS^{\rT})}}
    \sqrt{
    \Bra
     -S\Theta A^{-\rT}A^{-1} \Theta S^{\rT},
    \int_{\mR^d}
    f(w)%\pi^{-d/2}\sqrt{\det(\theta I_d + SPS^{\rT})}\re^{-\|w\|_{\theta I_d + SPS^{\rT}}^2}
    ww^{\rT}
    \rd w
    \Ket}\\
\label{dmudPsinorm}
    & = &
    \frac{\sqrt{2}\pi^{d/4}}{\sqrt[4]{\det(\theta I_d + SPS^{\rT})}}
    \sqrt{
    \Bra
     -S\Theta A^{-\rT}A^{-1} \Theta S^{\rT},
     (\theta I_d + SPS^{\rT})^{-1}
    \Ket}.
\end{eqnarray}
Here, use is made of an auxiliary zero-mean Gaussian PDF $f$ with the covariance matrix $\frac{1}{2}(\theta I_d + SPS^{\rT})^{-1}$ given by
\begin{equation}
\label{auxgauss}
    f(w)
    :=
    \pi^{-d/2}\sqrt{\det(\theta I_d + SPS^{\rT})}\re^{-\|w\|_{\theta I_d + SPS^{\rT}}^2},
    \qquad
    w\in \mR^d.
\end{equation}
Since the matrix $S$ is of full row rank, the controllability of the matrix pair $(A,B)$ in (\ref{A}) and (\ref{B}) (whereby $P$ in (\ref{Siginv}) satisfies $P\succ 0$) ensures that $SPS^{\rT}\succ 0$. In this case, the right-hand side of (\ref{dmudPsinorm}) has a limit as $\theta \to 0$, which corresponds to the standard $L^2$-norm $\|\d_{\Psi} \wt{\mu}\|_2$.
In view of (\ref{F''})--(\ref{G''}), the norms of the other Frechet derivatives $\d_{\Psi} \wt{P}$, $\d_{\Ups} \wt{\mu}$, $\d_{\Ups} \wt{P}$  are related in a similar fashion to the higher-order moments for the Gaussian PDF (\ref{auxgauss}), which can be computed in closed form with the aid of the Isserlis theorem \cite{I_1918,J_1997}.
The details of these calculations are cumbersome and are omitted for brevity (see also \cite{V_2018}).

Beyond the first two moments of the perturbed invariant quantum state discussed above, its sensitivity to the Weyl variations satisfying (\ref{PsiUpsgood}) can be quantified for the QCF correction $\wt{\Phi}_*$ by the corresponding $L^2$-induced  norm of the integral operator $(\Psi,\Ups)\mapsto \wt{\Phi}_*$ in (\ref{QCFvar}) of Theorem~\ref{th:QCFvar}. Since this norm is hard to compute, we will consider an upper  bound
\begin{eqnarray}
\nonumber
    \|\wt{\Phi}_*\|_2
     & := &
    \sqrt{\int_{\mR^n}|\wt{\Phi}_*(u)|^2 \rd u}\\
\label{HS}
      & \< &
    \sqrt{\sn F\sn_{-\theta}^2 + \sn G\sn_{-\theta}^2}
    \sqrt{\sn\Psi\sn_{\theta}^2 + \sn\Ups\sn_{\theta}^2},
\end{eqnarray}
which uses a weighted Hilbert-Schmidt norm \cite{DS_1963}  of this operator  involving the weighted $L^2$-norms of the influence functions $F$ and $G$ in (\ref{F0}) and (\ref{G0}):
\begin{eqnarray}
\label{Fnorm}
    \sn F\sn_{-\theta}
    & = &
    \sqrt{\int_{\mR^n \x \mR^d} |F(u,w)|^2 \re^{-\theta |w|^2}\rd u\rd w},\\
\label{Gnorm}
    \sn G\sn_{-\theta}
    & = &
    \sqrt{\int_{\mR^n \x \mR^d} |G(u,w)|^2 \re^{-\theta |w|^2}\rd u\rd w}.
\end{eqnarray}
Here, with a slight abuse of notation, the weighting is over the second arguments of the functions $F$ and $G$.
Up to a factor of $(2\pi)^{-n/2}$, the right-hand side of (\ref{HS}) is also an upper bound for the norm $\|\wt{\mho}_*\|_2 = (2\pi)^{-n/2}\|\wt{\Phi}_*\|_2$ of the first-order correction (\ref{mho'inf}) of the invariant QPDF (due to the Plancherel identity applied to $\wt{\mho}_*$ as the Fourier transform of $\wt{\Phi}_*$).

Although the evaluation of the integrals (\ref{F0}) and (\ref{G0}), required for (\ref{Fnorm}) and (\ref{Gnorm}), can be carried out numerically, we will consider upper bounds which give a qualitative insight into the influence functions  $F$ and $G$.
To this end (similarly to  \cite[Theorem 6]{VPJ_2017b}),  let $\lambda$ and $\Gamma$ be a positive scalar and a real positive definite symmetric matrix of order $n$  satisfying the algebraic Lyapunov inequality %(ALI)
\begin{equation}
\label{ALI}
    A\Gamma + \Gamma A^{\rT} + 2\lambda \Gamma \preccurlyeq 0.
\end{equation}
Such pairs $(\lambda,\Gamma)$ exist since the matrix $A$ is Hurwitz. Moreover,
(\ref{ALI}) has a positive definite solution $\Gamma$ for any $\lambda<-\max_{1\< k\< n}\Re \lambda_k$, where $\lambda_1, \ldots, \lambda_n$ are the eigenvalues  of $A$. In a generic case when $A$ is diagonalizable,  a particular choice of such a matrix $\Gamma$ can be carried out by using the eigenbasis of $A$.

%%%%%%%%%%%%%%%%%%%%%%%%%%%%%%%%%%%%%%%%%%%%%%%%%%%%%%%%%%%%%%%%%%%%%%%%%%%%%%%%%%%%%%%%%%%%%%%%%%%%%%%%%%%%
\begin{lem}
\label{lem:FGupp}
Suppose the unperturbed OQHO has a controllable matrix pair $(A,B)$ in (\ref{A}) and (\ref{B}), with $A$ Hurwitz, and its Hamiltonian and coupling operators are subject to the Weyl variations (\ref{h0eps})--(\ref{htilde}). Then the influence functions $F$ and $G$ in (\ref{F0}) and (\ref{G0}) admit the upper bounds
\begin{eqnarray}
\label{Fupp0}
  |F(u,w)|
  & \< &
  \frac{2\tau}{\lambda}
      \re^{
        -\frac{1}{2} (\|u\|_P^2+\|S^{\rT}w\|_P^2)
    }
      \Big(\re^{
                \|u\|_{\Gamma}\|P S^{\rT}w\|_{\Gamma^{-1}}
    }-1
    \Big),\\
\label{Gupp0}
      |G(u,w)| & \< &
    \frac{2}{\lambda}
          \re^{
        -\frac{1}{2} (\|u\|_P^2+\|S^{\rT}w\|_P^2)
    }
    \Big(
        \tau
        \big(
            \|D \Gamma^{-1/2}\|\|u\|_{\Gamma}
            + |D S^{\rT} w|
        \big)
        +
        2
        \frac{\|M\Theta \Gamma^{-1/2}\|}{\|P S^{\rT}w\|_{\Gamma^{-1}}}
    \Big)
    \Big(
    \re^{
                \|u\|_{\Gamma}\|P S^{\rT}w\|_{\Gamma^{-1}}
    }
    -1
    \Big)
\end{eqnarray}
for all $u \in \mR^n$ and $w\in \mR^d$. Here, $(\lambda,\Gamma)$ is any pair of a positive scalar $\lambda$ and a real positive definite symmetric matrix $\Gamma$ of order $n$  satisfying (\ref{ALI}), and
\begin{equation}
\label{tau}
  \tau :=
  \sqrt{
  \br
  \big(
    S\Theta \Gamma^{-1} \Theta S^{\rT}
    (SP\Gamma^{-1} PS^{\rT})^{-1}
  \big)},
\end{equation}
where $\br(\cdot)$ denotes the spectral radius of a square matrix.
\hfill$\square$
\end{lem}
%%%%%%%%%%%%%%%%%%%%%%%%%%%%%%%%%%%%%%%%%%%%%%%%%%%%%%%%%%%%%%%%%%%%%%%%%%%%%%%%%%%%%%%%%%%%%%%%%%%
\begin{proof}
For any pair $(\lambda,\Gamma)$, satisfying  (\ref{ALI}), application of the Gronwall-Bellman lemma argument in the form of the differential matrix inequality
$
    \big(\re^{2\lambda t}\re^{t A} \Gamma \re^{tA^{\rT}}\big)^{^\centerdot}
    =
    \re^{2\lambda t}
    \re^{t A} (A\Gamma + \Gamma A^{\rT} + 2\lambda\Gamma) \re^{t A^{\rT}}
    \preccurlyeq
    0
$
leads to $    \re^{t A} \Gamma \re^{t A^{\rT}}
    \preccurlyeq
    \re^{-2\lambda t}\Gamma
$ for all $t\>0$, which is equivalent to the contraction
\begin{equation}
\label{contr}
    \|\Gamma^{-1/2} \re^{t A} \sqrt{\Gamma}\|
    \<
    \re^{-\lambda t},
    \qquad
    t\>0,
\end{equation}
in terms of the matrix operator norm.
By combining (\ref{F0}) with submultiplicativity of the operator norm and using (\ref{contr}), it follows that
\begin{eqnarray}
\nonumber
  |F(u,w)|
  & \< &
  2
      \re^{
        -\frac{1}{2} (\|u\|_P^2+\|S^{\rT}w\|_P^2)
    }
  \int_0^{+\infty}
  |\sin(u^{\rT}\re^{tA} \Theta S^{\rT} w)|
      \re^{
        -
        u^{\rT}\re^{tA} P S^{\rT}w
    }
  \rd t\\
\nonumber
  & \< &
  2
      \re^{
        -\frac{1}{2} (\|u\|_P^2+\|S^{\rT}w\|_P^2)
    }
  \int_0^{+\infty}
  |u^{\rT}\sqrt{\Gamma}\Gamma^{-1/2}\re^{tA} \sqrt{\Gamma} \Gamma^{-1/2}\Theta S^{\rT} w|
      \re^{
                |u^{\rT}\sqrt{\Gamma}\Gamma^{-1/2}\re^{tA} \sqrt{\Gamma} \Gamma^{-1/2}P S^{\rT}w|
    }
  \rd t\\
\nonumber
  & \< &
  2
      \re^{
        -\frac{1}{2} (\|u\|_P^2+\|S^{\rT}w\|_P^2)
    }
      \|u\|_{\Gamma}\|\Theta S^{\rT} w\|_{\Gamma^{-1}}
  \int_0^{+\infty}
  \|\Gamma^{-1/2}\re^{tA} \sqrt{\Gamma}\|
      \re^{
                \|u\|_{\Gamma} \|\Gamma^{-1/2}\re^{tA} \sqrt{\Gamma}\| \|P S^{\rT}w\|_{\Gamma^{-1}}
    }
  \rd t\\
\nonumber
  & \< &
  2
      \re^{
        -\frac{1}{2} (\|u\|_P^2+\|S^{\rT}w\|_P^2)
    }
      \|u\|_{\Gamma}\|\Theta S^{\rT} w\|_{\Gamma^{-1}}
  \int_0^{+\infty}
  \re^{-\lambda t}
      \re^{
                \|u\|_{\Gamma} \|P S^{\rT}w\|_{\Gamma^{-1}} \re^{-\lambda t}
    }
  \rd t\\
\nonumber
  & = &
  \frac{2}{\lambda}
      \re^{
        -\frac{1}{2} (\|u\|_P^2+\|S^{\rT}w\|_P^2)
    }
      \|\Theta S^{\rT} w\|_{\Gamma^{-1}}
  \int_0^{\|u\|_{\Gamma}}
      \re^{
                \|P S^{\rT}w\|_{\Gamma^{-1}} z
    }
  \rd z\\
\nonumber
  & = &
  \frac{2}{\lambda}
      \re^{
        -\frac{1}{2} (\|u\|_P^2+\|S^{\rT}w\|_P^2)
    }
      \frac{\|\Theta S^{\rT} w\|_{\Gamma^{-1}}}{\|P S^{\rT}w\|_{\Gamma^{-1}}}
      \Big(\re^{
                \|u\|_{\Gamma}\|P S^{\rT}w\|_{\Gamma^{-1}}
    }-1
    \Big)\\
\label{Fupp}
  & \< &
  \frac{2\tau}{\lambda}
      \re^{
        -\frac{1}{2} (\|u\|_P^2+\|S^{\rT}w\|_P^2)
    }
      \Big(\re^{
                \|u\|_{\Gamma}\|P S^{\rT}w\|_{\Gamma^{-1}}
    }-1
    \Big),
\end{eqnarray}
for all $u \in \mR^n$ and $w\in \mR^d$, thus establishing (\ref{Fupp0}), where the quantity
$$
    \tau
    :=
    \max_{w\in \mR^d\setminus \{0\}}
    \frac{\|\Theta S^{\rT} w\|_{\Gamma^{-1}}}{\|P S^{\rT}w\|_{\Gamma^{-1}}}
    =
    \max_{w\in \mR^d\setminus \{0\}}
    \frac{\|\Theta S^{\rT} w\|_{\Gamma^{-1}}}{\big|\sqrt{SP\Gamma^{-1} PS^{\rT}}w\big|}
$$
is given by (\ref{tau}). Here, the positive definiteness of the matrix  $SP\Gamma^{-1} PS^{\rT}$ is ensured by $S$ being of full row rank and by $P\succ 0$ (due to controllability of the pair $(A,B)$) together with $\Gamma\succ 0$.
Note that the double exponential function in the second last integral of (\ref{Fupp}) comes from (\ref{contr}) and is circumvented by changing the integration variable.
By a similar reasoning, a combination of (\ref{G0}) with (\ref{contr}) leads to the upper bound
\begin{eqnarray}
\nonumber
  |G(u,w)|
  &\< &
    2
          \re^{
        -\frac{1}{2} (\|u\|_P^2+\|S^{\rT}w\|_P^2)
    }
    \int_0^{+\infty}
    \big(
        |\sin(u^{\rT}\re^{tA}\Theta S^{\rT}w)| |D (\re^{tA^{\rT}}u+S^{\rT}w)|
        +
        |K(\re^{tA^{\rT}}u,S^{\rT}w)M\Theta \re^{tA^{\rT}}u|
    \big)
        \re^{
        -
        u^{\rT}\re^{tA} P S^{\rT}w
    }
    \rd t\\
\nonumber
   & \< &
    2
          \re^{
        -\frac{1}{2} (\|u\|_P^2+\|S^{\rT}w\|_P^2)
    }
    \|u\|_{\Gamma}
    \int_0^{+\infty}
    \big(
        \|\Theta S^{\rT}w\|_{\Gamma^{-1}}
        \big(
            \|D \Gamma^{-1/2}\|\|u\|_{\Gamma}\re^{-\lambda t}
            + |D S^{\rT} w|
        \big)
        +
        2
        \|M\Theta \Gamma^{-1/2}\|
    \big)\\
\nonumber
    & &\x
    \re^{-\lambda t}
    \re^{
                \|u\|_{\Gamma} \|P S^{\rT}w\|_{\Gamma^{-1}} \re^{-\lambda t}
    }\rd t\\
\nonumber
    &= &
    \frac{2}{\lambda}
          \re^{
        -\frac{1}{2} (\|u\|_P^2+\|S^{\rT}w\|_P^2)
    }
    \int_0^{\|u\|_{\Gamma}}
    \big(
        \|\Theta S^{\rT}w\|_{\Gamma^{-1}}
        \big(
            \|D \Gamma^{-1/2}\|z
            + |D S^{\rT} w|
        \big)
        +
        2
        \|M\Theta \Gamma^{-1/2}\|
    \big)
    \re^{
                \|P S^{\rT}w\|_{\Gamma^{-1}} z
    }\rd z\\
\label{Gupp}
    &\< &
    \frac{2}{\lambda}
          \re^{
        -\frac{1}{2} (\|u\|_P^2+\|S^{\rT}w\|_P^2)
    }
    \Big(
        \tau
        \big(
            \|D \Gamma^{-1/2}\|\|u\|_{\Gamma}
            + |D S^{\rT} w|
        \big)
        +
        2
        \frac{\|M\Theta \Gamma^{-1/2}\|}{\|P S^{\rT}w\|_{\Gamma^{-1}}}
    \Big)
    \Big(
    \re^{
                \|u\|_{\Gamma}\|P S^{\rT}w\|_{\Gamma^{-1}}
    }
    -1
    \Big),
\end{eqnarray}
which proves (\ref{Gupp0}). In (\ref{Gupp}), we have used the property $\|J\| = 1$ of the matrix $J$ from (\ref{J}), whereby the function $K$ in (\ref{K}) satisfies $\|K(u,v)\| \< 2$
for all $u,v\in \mR^n$.
\end{proof}
%%%%%%%%%%%%%%%%%%%%%%%%%%%%%%%%%%%%%%%%%%%%%%%%%%%%%%%%%%%%%%%%%%%%%%%%%%%%%%%%%%%%%%%%%%%%%%%%%%%

Since the upper bounds in (\ref{Fupp0}) and (\ref{Gupp0}) involve the quadratic-exponential factors which are proportional to Gaussian PDFs, they lead to upper bounds for the weighted $L^2$-norms of the influence functions  $F$ and $G$ in (\ref{Fnorm}) and (\ref{Gnorm}) which can be evaluated numerically by Monte-Carlo methods. More precisely, a combination of (\ref{Fnorm}) with (\ref{Fupp0}) leads to
\begin{eqnarray}
\nonumber
    \sn F\sn_{-\theta}^2
    & \< &
    \Big(\frac{2\tau}{\lambda}\Big)^2
    \int_{\mR^n \x \mR^d}
      \re^{
        -\|u\|_P^2-\|w\|_{\theta I_d + SPS^{\rT}}^2
    }
      \Big(\re^{
                \|u\|_{\Gamma}\|P S^{\rT}w\|_{\Gamma^{-1}}
    }-1
    \Big)^2 \rd u\rd w\\
\label{Fnormupp}
    & = &
    \Big(\frac{2\tau}{\lambda}\Big)^2
    \frac{\pi^{(n+d)/2}}{\sqrt{\det P \det(\theta I_d + SPS^{\rT})}}
    \bE
    \Big(
    \big(\re^{
                \|\nu\|_{\Gamma}\|P S^{\rT}\omega\|_{\Gamma^{-1}}
    }-1
    \big)^2\Big).
\end{eqnarray}
Here, the expectation is over
independent zero-mean Gaussian random vectors $\nu$ and $\omega$  with values in $\mR^n$ and $\mR^d$ and   covariance matrices $\frac{1}{2}P^{-1}$ and $\frac{1}{2}(\theta I_d + SPS^{\rT})^{-1}$, respectively (so that $\omega$ has the PDF $f$ given by (\ref{auxgauss})). By a similar reasoning, (\ref{Gnorm}) and (\ref{Gupp0}) imply that
\begin{eqnarray}
\nonumber
    \sn G\sn_{-\theta}^2
    &\< &
    \Big(\frac{2}{\lambda}\Big)^2
    \frac{\pi^{(n+d)/2}}{\sqrt{\det P \det(\theta I_d + SPS^{\rT})}}\\
\label{Gnormupp}
    && \x \bE
    \Big(
    \Big(
    \Big(
        \tau
        \big(
            \|D \Gamma^{-1/2}\|\|\nu\|_{\Gamma}
            + |D S^{\rT} \omega|
        \big)
        +
        2
        \frac{\|M\Theta \Gamma^{-1/2}\|}{\|P S^{\rT}\omega\|_{\Gamma^{-1}}}
    \Big)
    \Big(
    \re^{
                \|u\|_{\Gamma}\|P S^{\rT}\omega\|_{\Gamma^{-1}}
    }
    -1
    \Big)
    \Big)^2\Big).
\end{eqnarray}
Therefore, the right-hand sides of (\ref{Fnormupp}) and (\ref{Gnormupp}) can be estimated via the Monte-Carlo method by sampling from the classical Gaussian distributions.

%
%%%%%%%%%%%%%%%%%%%%%%%%%%%%%%%%%%%%%%%%%%%%%%%%%%%%%%%%%%%%%%%%%%%%%%%%%%%%%%%%%%%%%%%%%%%%%%%%%%%%
\section{NUMERICAL EXAMPLES OF INVARIANT STATE PERTURBATION ANALYSIS}\label{sec:numer}
%%%%%%%%%%%%%%%%%%%%%%%%%%%%%%%%%%%%%%%%%%%%%%%%%%%%%%%%%%%%%%%%%%%%%%%%%%%%%%%%%%%%%%%%%%%%%%%%%%%%

As an illustration  of the above results, we will now provide two numerical examples of infinitesimal perturbation analysis of invariant quantum states under the Weyl variations of the energy operators. First, consider a one-mode OQHO with $n=2$ system variables consisting of the position $q$ and momentum $p$ operators (see Example~1) driven by $m=2$ external fields.
The energy and coupling matrices $R$ and $M$  of the OQHO in (\ref{hhOQHO1}) and (\ref{hhOQHO2}) were randomly generated so as to make the matrix $A$  in (\ref{A}) Hurwitz:
\begin{equation}
\label{RMex1}
    R
     =
    {\small
    \left[\begin{array}{cc}
    1.5803  &       0\\
         0  &  0.7490
    \end{array}
    \right]},
    \qquad
    M
    =
    {\small
    \left[
    \begin{array}{cc}
   -0.1765  & -1.3320\\
    0.7914  & -2.3299
    \end{array}
    \right]}
\end{equation}
(the stiffness and mass parameters in (\ref{Rqp}) are $\sK = 1.5803$ and $\sM = 1.3351$).
The resulting state-space matrices in (\ref{A})--(\ref{C}) are
$$
    A
    =
    {\small
    \left[
    \begin{array}{cc}
   -1.4654  &  0.7490\\
   -1.5803  & -1.4654
    \end{array}
    \right]},
    \qquad
B =
    {\small
    \left[
    \begin{array}{cc}
   -1.3320 &  -2.3299\\
    0.1765 &  -0.7914
    \end{array}
    \right]},
    \qquad
C =
    {\small
    \left[
    \begin{array}{cc}
    1.5828 &  -4.6598\\
    0.3530 &   2.6640
    \end{array}
    \right]}
$$
(the eigenvalues of $A$ are $-1.4654 \pm 1.0880i$). The real part of the quantum covariance matrix of the invariant Gaussian state of the OQHO is
$$
P =
{\small
\left[
\begin{array}{cc}
    2.2207 &  -0.4635\\
   -0.4635 &   0.7241
\end{array}
\right]}.
$$
Following Example~4 (with the coupling operators and the kinetic energy part of the Hamiltonian remaining unperturbed), consider the QCF correction in response to a perturbation of the nominal quadratic potential $\frac{1}{2}\sK q^2$ in the form of a Gaussian-shaped well potential $\phi$ in (\ref{phiMorse}) with parameters
\begin{equation}
\label{well}
    \alpha = -146.0546,
    \qquad
    \gamma = 3.1641,
    \qquad
    \Lambda = 0.1589.
\end{equation}
The stiffness of this well (at its centre $\gamma$) is $-\alpha\Lambda^{-1} = 919.0101$. The perturbed potential $\frac{1}{2}\sK q^2 + \eps \phi(q)$  (with $\eps = 0.03$) is shown in Fig.~\ref{fig:pot}.
%%%%%%%%%%%%%%%%%%%%%%%%%%%%%%%%%%%%%%%%%%%%%%%%%%%%%%%%%%%%%%%%%%%%%%%%%%%%%%%%%%%%%%%%%%%%%%%%%%%%
\begin{figure}[thpb]
      \centering
      \includegraphics[width=8cm]{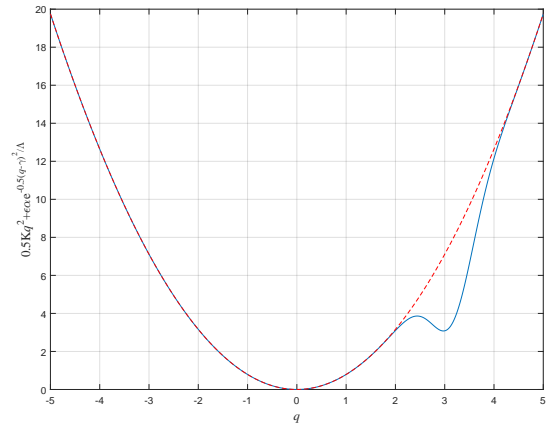}
      \caption{The perturbed potential $\frac{1}{2}\sK q^2 + \eps \alpha \re^{-\frac{1}{2}(q-\gamma)^2/\Lambda}$ (solid-line curve), consisting of the nominal quadratic part (dashed-line parabola) with parameters from (\ref{RMex1}) and a Gaussian-shaped well potential with parameters (\ref{well}) and $\eps = 0.03$.}
      \label{fig:pot}
   \end{figure}
%%%%%%%%%%%%%%%%%%%%%%%%%%%%%%%%%%%%%%%%%%%%%%%%%%%%%%%%%%%%%%%%%%%%%%%%%%%%%%%%%%%%%%%%%%%%%%%%%%%%
The real and imaginary parts of the corresponding QCF correction $\wt{\Phi}_*$, evaluated by numerical integration of (\ref{QCFcorr}) using (\ref{Xi})--(\ref{sig+-}) with $\Xi = 2.3797$, are depicted in Fig.~\ref{fig:QCFcorr}.
%%%%%%%%%%%%%%%%%%%%%%%%%%%%%%%%%%%%%%%%%%%%%%%%%%%%%%%%%%%%%%%%%%%%%%%%%%%%%%%%%%%%%%%%%%%%%%%%%%%%
\begin{figure}[thpb]
      \centering
      \includegraphics[width=8cm]{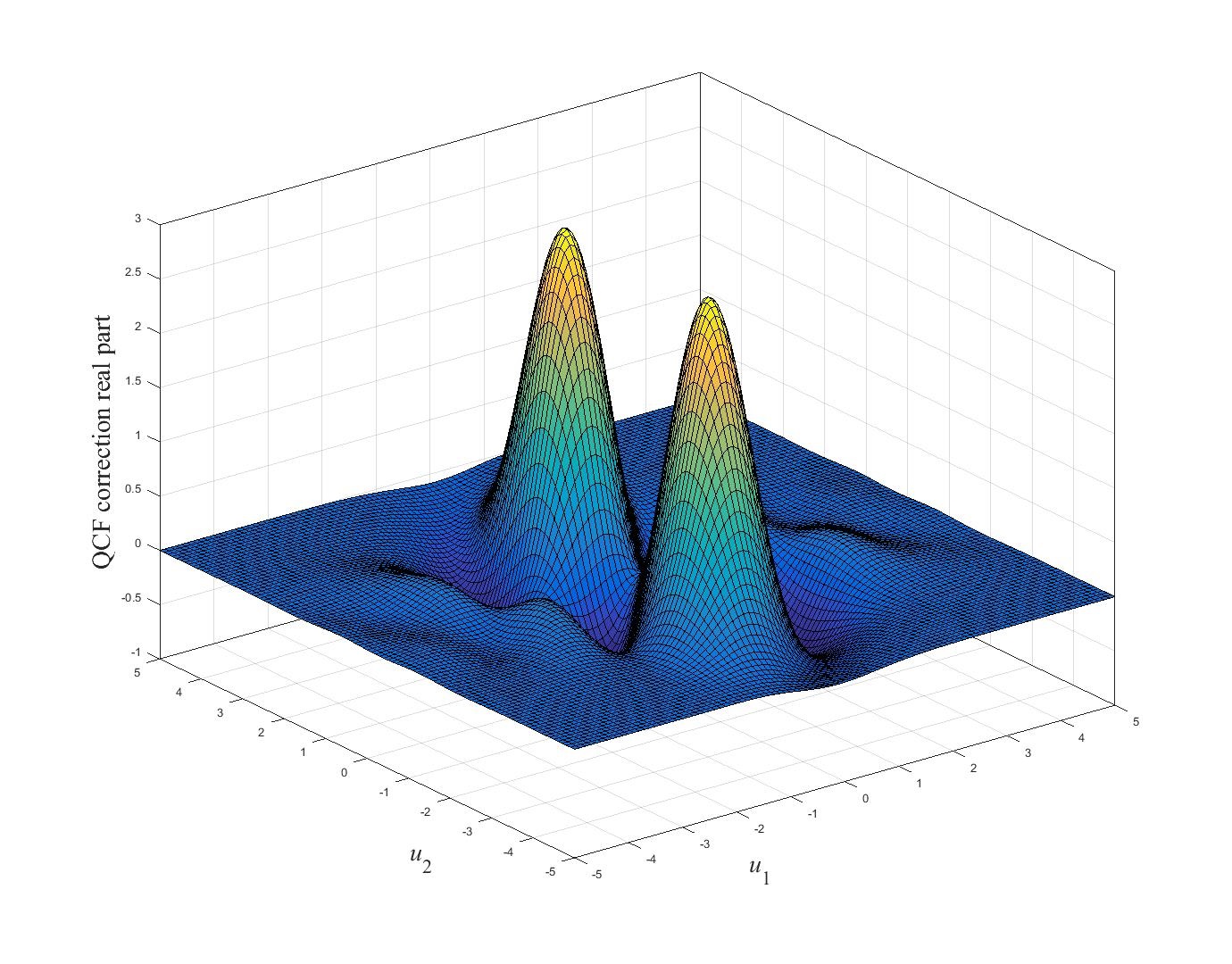}
      \includegraphics[width=8cm]{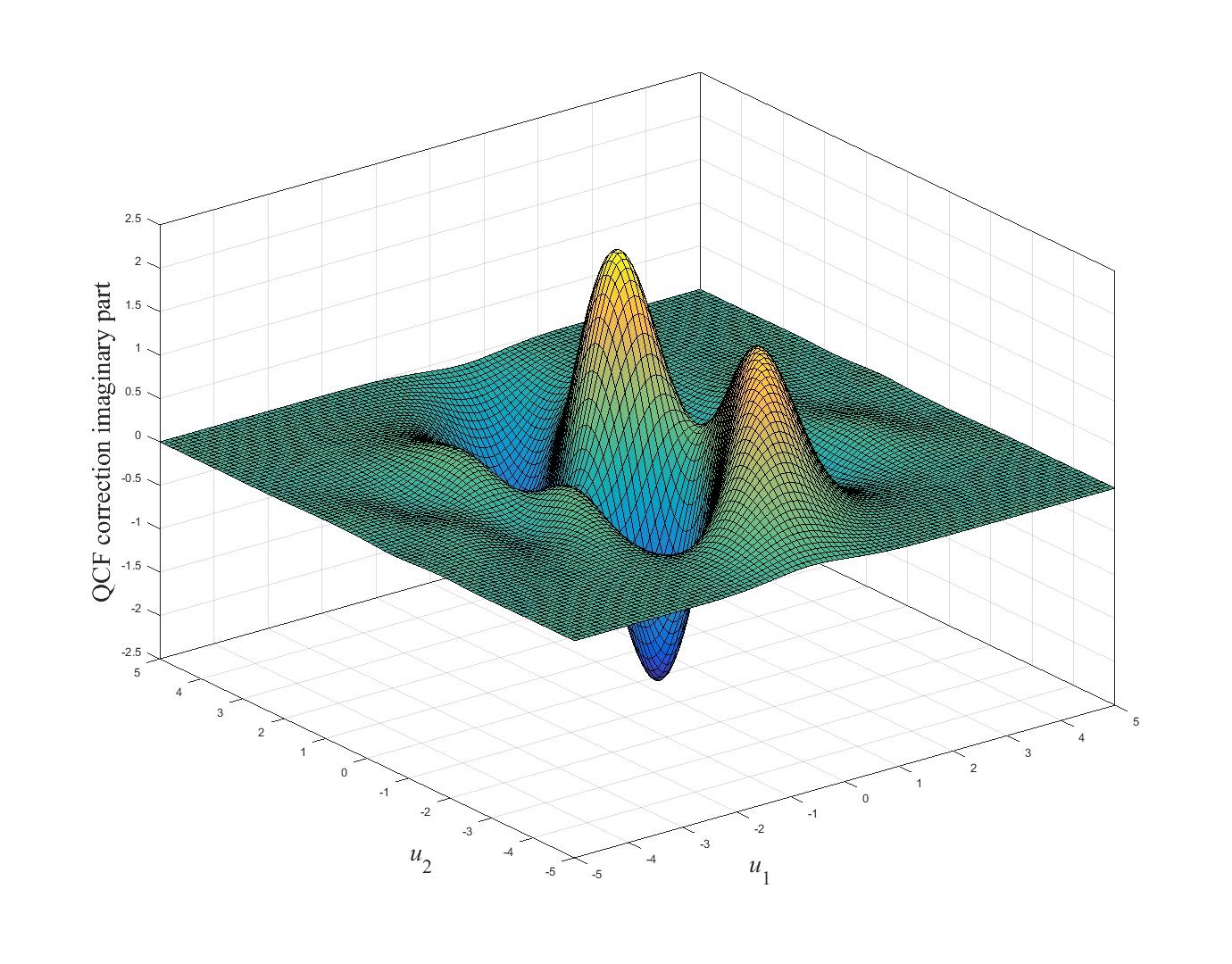}
      \caption{The real and imaginary parts of the first-order QCF correction $\wt{\Phi}_*$ in (\ref{QCFcorr}) for the perturbed one-mode oscillator. }
      \label{fig:QCFcorr}
   \end{figure}
%%%%%%%%%%%%%%%%%%%%%%%%%%%%%%%%%%%%%%%%%%%%%%%%%%%%%%%%%%%%%%%%%%%%%%%%%%%%%%%%%%%%%%%%%%%%%%%%%%%%
The QPDF correction $\wt{\mho}_*$, obtained by the Fourier transform of $\wt{\Phi}_*$ in (\ref{mho'inf}), is shown in Fig.~\ref{fig:QPDFcorr}.
%%%%%%%%%%%%%%%%%%%%%%%%%%%%%%%%%%%%%%%%%%%%%%%%%%%%%%%%%%%%%%%%%%%%%%%%%%%%%%%%%%%%%%%%%%%%%%%%%%%%
\begin{figure}[thpb]
      \centering
      \includegraphics[width=8cm]{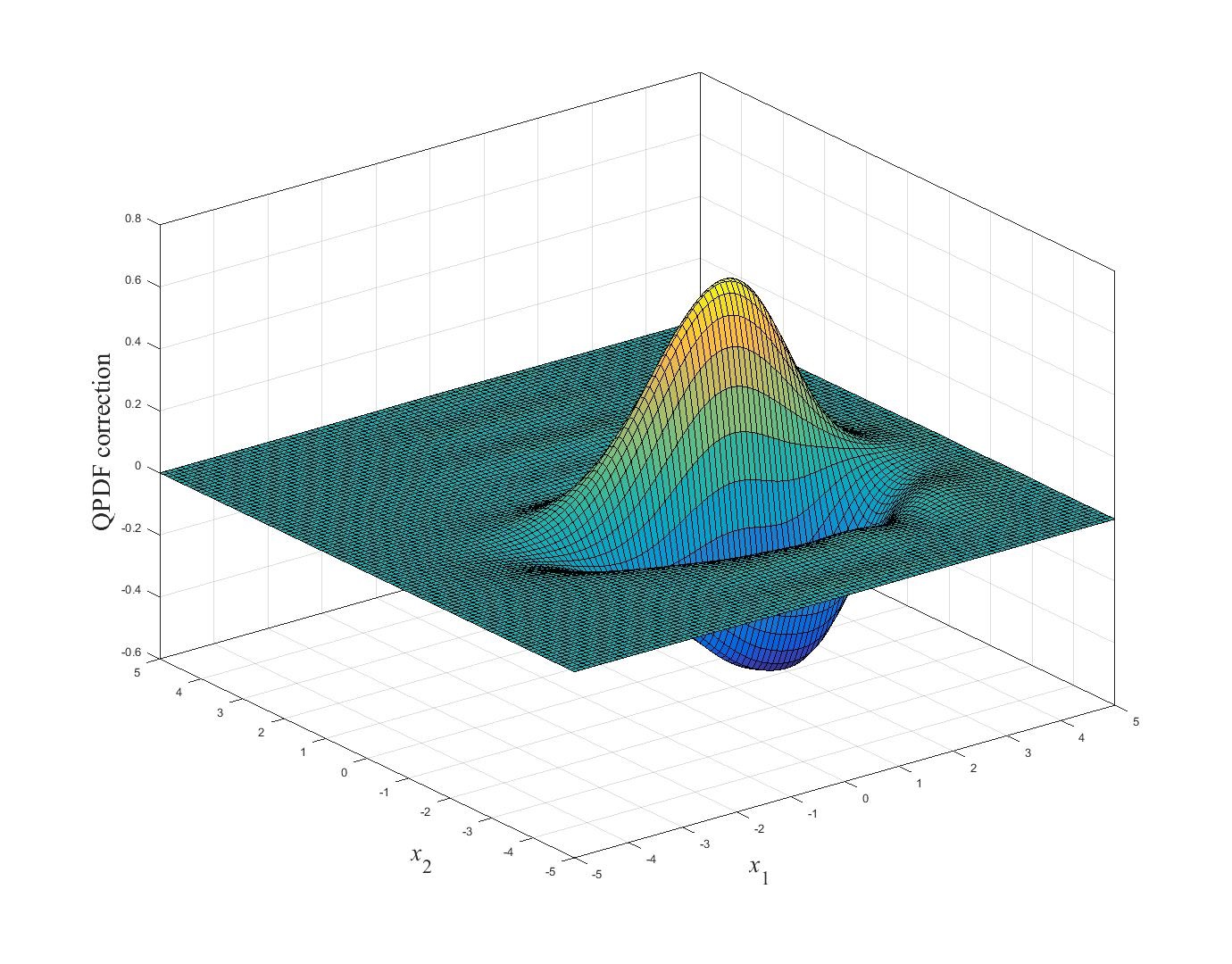}
      \caption{The first-order QPDF correction $\wt{\mho}_*$ in (\ref{mho'inf}) (as a function of the position $x_1$ and momentum $x_2$ coordinates) for the perturbed one-mode oscillator. }
      \label{fig:QPDFcorr}
   \end{figure}
%%%%%%%%%%%%%%%%%%%%%%%%%%%%%%%%%%%%%%%%%%%%%%%%%%%%%%%%%%%%%%%%%%%%%%%%%%%%%%%%%%%%%%%%%%%%%%%%%%%%
The corresponding marginal QPDF correction $\int_{-\infty}^{+\infty}\wt{\mho}_*(x_1,x_2)\rd x_2$ for the position variable of the perturbed OQHO is depicted in Fig.~\ref{fig:QPDFcorr_pos}.
%%%%%%%%%%%%%%%%%%%%%%%%%%%%%%%%%%%%%%%%%%%%%%%%%%%%%%%%%%%%%%%%%%%%%%%%%%%%%%%%%%%%%%%%%%%%%%%%%%%%
\begin{figure}[thpb]
      \centering
      \includegraphics[width=8cm]{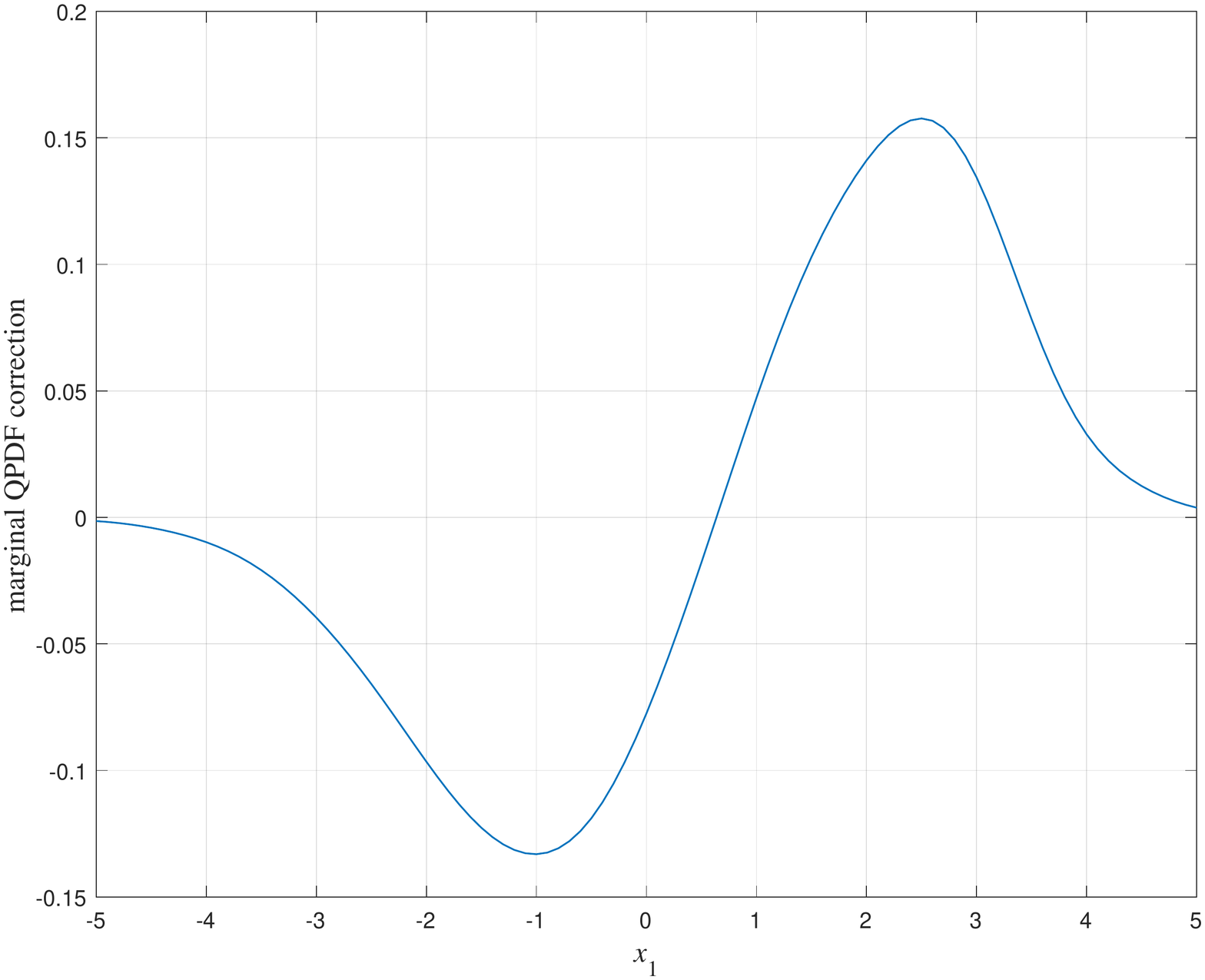}
      \caption{The first-order correction $\int_{-\infty}^{+\infty}\wt{\mho}_*(x_1,x_2)\rd x_2$ to the marginal QPDF of the position variable of the perturbed one-mode oscillator. }
      \label{fig:QPDFcorr_pos}
   \end{figure}
%%%%%%%%%%%%%%%%%%%%%%%%%%%%%%%%%%%%%%%%%%%%%%%%%%%%%%%%%%%%%%%%%%%%%%%%%%%%%%%%%%%%%%%%%%%%%%%%%%%%
Although the perturbed marginal QPDF of the position variable is no longer Gaussian, it remains nonnegative everywhere on the real line (as does the classical probability distribution of any self-adjoint operator). The correction in Fig.~\ref{fig:QPDFcorr_pos} increases the value of the perturbed QPDF in the vicinity of the secondary potential well in Fig.~\ref{fig:pot} at $x_1\approx 3$ by ``taking'' some probability from negative values of the position coordinate $x_1$ (in order to satisfy the normalization constraint similar to (\ref{Phimhonorm})). Also, this correction contributes a positive shift  $\int_{-\infty}^{+\infty}x_1\wt{\mho}_*(x_1,x_2)\rd x_1 \rd x_2 = 1.3677$  to the mean value of the position from $0$ towards the well.

As another example, consider a two-mode OQHO with $n=4$ system variables consisting of the conjugate positions $q_1, q_2$ and momenta $p_1, p_2$  driven by $m=4$ fields.
The energy and coupling matrices of the OQHO (also randomly generated so as to make the matrix $A$  in (\ref{A}) Hurwitz) are:
\begin{equation}
\label{RMex2}
    R
     =
    {\small
    \left[
    \begin{array}{cccc}
    2.5542 &  -2.3651 &        0 &        0\\
   -2.3651 &   2.6995 &        0 &        0\\
         0 &        0 &   0.9306 &  -1.4504\\
         0 &        0 &  -1.4504 &   7.4900
    \end{array}
    \right]},
    \qquad
    M
    =
    {\small
    \left[
    \begin{array}{cccc}
    0.3021 &   1.1784 &   0.0313 &  -1.4647\\
    0.0131 &  -0.2981 &   1.5002 &   0.5361\\
   -0.0110 &  -0.0418 &  -1.1125 &   1.5380\\
   -0.7233 &  -1.0734 &   0.7212 &   0.1241
    \end{array}
    \right]}.
\end{equation}
The spectrum of the resulting matrix
$$
    A =
    {\small
    \left[
    \begin{array}{cccc}
   -0.7588 &  -0.0857 &   0.9306 &  -3.2322\\
   -0.8379 &  -2.2896 &   0.3314 &   7.4900\\
   -2.5542 &   2.5944 &  -0.7588 &  -0.8379\\
    2.1358 &  -2.6995 &  -0.0857 &  -2.2896
    \end{array}
    \right]}
$$
is $\{  -1.7804 \pm 5.1767i,   -1.2680 \pm 0.8111i\}$. Now, let the potential energy part of the system Hamiltonian be subject to the Weyl variations whose strength functions $\Psi$ satisfy (\ref{PsiUpsgood}). With the matrix $S$ given by (\ref{SId}) with $d=2$, the norm $\sn\d_{\Psi} \wt{\mu}\sn_{-\theta}$ in (\ref{dmudPsinorm}), which quantifies the sensitivity of the perturbed invariant mean vector, is shown in Fig.~\ref{fig:dmudPsi}.
%%%%%%%%%%%%%%%%%%%%%%%%%%%%%%%%%%%%%%%%%%%%%%%%%%%%%%%%%%%%%%%%%%%%%%%%%%%%%%%%%%%%%%%%%%%%%%%%%%%%
\begin{figure}[thpb]
      \centering
      \includegraphics[width=8cm]{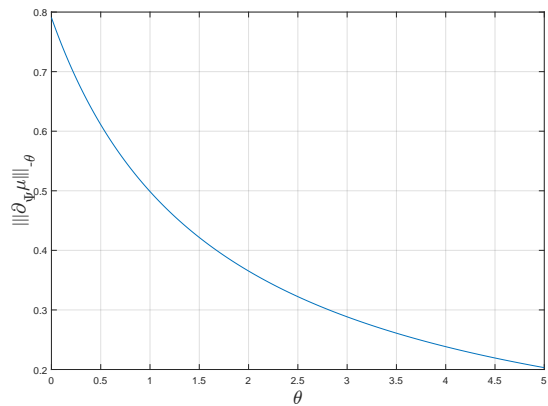}%\hskip-5mm
      \caption{The sensitivity of the mean vector of the system variables for the perturbed invariant state with respect to the Weyl variations in the potential energy in terms of the norm $\sn\d_{\Psi} \wt{\mu}\sn_{-\theta}$ in (\ref{dmudPsinorm}) for the two-mode OQHO with parameters (\ref{RMex2}). }
      \label{fig:dmudPsi}
   \end{figure}
%%%%%%%%%%%%%%%%%%%%%%%%%%%%%%%%%%%%%%%%%%%%%%%%%%%%%%%%%%%%%%%%%%%%%%%%%%%%%%%%%%%%%%%%%%%%%%%%%%%%

%%%%%%%%%%%%%%%%%%%%%%%%%%%%%%%%%%%%%%%%%%%%%%%%%%%%%%%%%%%%%%%%%%%%%%%%%%%%%%%%%%%%%%%%%%%%%%%%%%%%%%%%
\section{Conclusion}\label{sec:conc}
%%%%%%%%%%%%%%%%%%%%%%%%%%%%%%%%%%%%%%%%%%%%%%%%%%%%%%%%%%%%%%%%%%%%%%%%%%%%%%%%%%%%%%%%%%%%%%%%%%%%%%%%

For a class of open quantum harmonic oscillators, whose linear-quadratic coupling and Hamiltonian operators are subject to the Weyl variations leading to nonlinearities in the governing QSDE and non-Gaussian state dynamics, we have obtained the first-order correction terms for the QCF and QPDF of the invariant state. We have also discussed the norms of the linear operators which relate these corrections to the Weyl variations. These infinitesimal perturbation analysis results can find applications to the Gaussian state generation and approximate invariant state computation taking into account inaccuracies and uncertainties in the implementation of such systems.  Another possible direction of the research, reported in this paper,  is an extension of these ideas to more general states of the external fields (such as nonvacuum Gaussian states \cite{N_2014}).

\end{document}